\documentclass[authoryear,12pt]{elsarticle}
\usepackage{mathtools}
\usepackage{graphicx}
\usepackage{amsmath}
\usepackage{amsthm}
\usepackage{amssymb}
\usepackage{caption}
\usepackage{subcaption}

\usepackage{placeins}
\usepackage{siunitx}
\usepackage{soul}

\usepackage{epstopdf}
\usepackage[T1]{fontenc}				
\usepackage{booktabs}					
\usepackage{setspace}
\usepackage[english]{babel}
\usepackage{float}
\usepackage[nomarkers]{endfloat} 
\usepackage[ruled,vlined]{algorithm2e}
\usepackage{algpseudocode}
\usepackage{ragged2e}

\theoremstyle{plain}
\newtheorem{thm}{Theorem}

\newtheorem{corollary}{Corollary}
\theoremstyle{remark}
\newtheorem{rem}{Remark}

\newtheorem{definition}{Definition}
\usepackage{graphicx}
\newtheorem{lemma}[thm]{Lemma}

\newcommand\norm[1]{\left\lVert#1\right\rVert}

\renewcommand{\refname}{Literature Cited}

\allowdisplaybreaks

\journal{Elsevier}

\begin{document}
	
	\baselineskip 24pt
	
	\begin{frontmatter}
		
                \author[mymainaddress]{Wallace Tan Gian Yion}
			\author[mymainaddress]{Zhe Wu\corref{mycorrespondingauthor}}
			\cortext[mycorrespondingauthor]{Corresponding author: E-mail: wuzhe@nus.edu.sg.}
			\address[mymainaddress]{Department of Chemical and Biomolecular Engineering, National University of Singapore, 117585, Singapore.}
		
		\title{Robust Machine Learning Modeling for Predictive Control Using Lipschitz-Constrained Neural Networks}

	\begin{abstract}
		Neural networks (NNs) have emerged as a state-of-the-art method for modeling nonlinear systems in model predictive control (MPC). However, the robustness of NNs, in terms of sensitivity to small input perturbations, remains a critical challenge for practical applications. To address this, we develop Lipschitz-Constrained Neural Networks (LCNNs) for modeling nonlinear systems and derive rigorous theoretical results to demonstrate their effectiveness in approximating Lipschitz functions, reducing input sensitivity, and preventing over-fitting. Specifically, we first prove a universal approximation theorem to show that LCNNs using SpectralDense layers can approximate any 1-Lipschitz target function. Then, we prove a probabilistic generalization error bound for LCNNs using SpectralDense layers by using their empirical Rademacher complexity. Finally, the LCNNs are incorporated into the MPC scheme, and a chemical process example is utilized to show that LCNN-based MPC outperforms MPC using conventional feedforward NNs in the presence of training data noise.
	\end{abstract}
	
\begin{keyword}
   Lipschitz-Constrained Neural Networks; Robust Machine Learning Model; Generalization Error;   Model Predictive Control; Neural Network Sensitivity; Over-fitting  
\end{keyword}
\end{frontmatter}

\section{Introduction}
 Model Predictive Control (MPC) is an advanced optimization-based control strategy for various chemical engineering processes, such as batch crystallization processes (\cite{kwon_modeling_2013,kwon_protein_2014}) and  continuous stirred tank reactors (CSTRs) (\cite{chen1995nonlinear,wu2001lmi}).  Machine learning (ML) techniques such as artificial neural networks (ANN) have been utilized to develop prediction models that will be incorporated into the design of MPC. Among various ML-based MPC schemes, an accurate and robust process model for prediction has always been one of the key components that ensures desired closed-loop performance.

Despite the success of ANNs in modeling complex nonlinear systems, one prominent issue is that they could potentially be sensitive to small perturbations in input features. Such sensitivity issues arise naturally when a small change in the input (e.g., due to data noise, perturbation, or artificially generated adversarial inputs) could result in a drastic change in the output. For example, in classification problems, adversarial input perturbations have been shown to lead to a large variation in NN output, thereby leading to misclassified results (\cite{szegedy2013intriguing}). Additionally, \cite{balda2019perturbation} proposed a novel approach for constructing adversarial examples for regression problems using perturbation analysis for the underlying learning algorithm for the neural network. Since the lack of robustness of NNs could adversely affect the prediction accuracy in ML-based MPC, it is important to address the sensitivity issue for the implementation of NNs in performance-critical applications. 

To mitigate this issue, adversarial training has been adopted as one of the most effective approaches to train NNs against adversarial examples (\cite{szegedy2013intriguing}). For example, \cite{shaham2015understanding} proposed a robust optimization framework with a different loss function in the optimization process that aims to search for adversarial examples near the vicinity of each training data point.  This approach has been empirically shown to be effective against adversarial attacks by both \cite{shaham2015understanding} and  \cite{mkadry2017towards}. However, there is a lack of a provably performance guarantee on the input sensitivity of neural networks for this approach. 

Therefore, a type of neural network with a fixed Lipschitz constant termed Lipschitz-Constrained Neural Networks (LCNNs) has received an emerging amount of interest in recent years (\cite{baldi2013understanding,anil2019sorting}). One immediate way to control the Lipschitz constant of a neural network is to bound the norms of the weight matrices, and use activation functions with bounded derivatives (e.g., ReLU activation function). However, it is demonstrated in \cite{anil2019sorting} that this method substantially limits the  model capacity of the NNs if component-wise activation functions are used. A recent breakthrough using a special activation function termed GroupSort (\cite{anil2019sorting}) significantly increases the expressive power of the NNs. Specifically, the GroupSort activation function enables NNs with norm-constrained GroupSort architectures to serve as universal approximators of 1-Lipschitz functions. Using the restricted Stone-Weierstrass theorem, \cite{anil2019sorting} proved the universal approximation theorem for GroupSort feedforward neural networks with appropriate assumptions on the norms of the weight matrices. However, since the proof in \cite{anil2019sorting} is based on the $\infty$-norm of matrices, the way to demonstrate the universal approximation property of LCNNs using the spectral norm (i.e., $\ell_2$
matrix norm) remains an open question.

The second prominent issue in the development of ANNs is over-fitting, where the networks perform very well on the training data but fail to predict accurately on the test data, which results in a high generalization error. One possible reason for over-fitting is that the training data has noise that negatively impacts learning performance (\cite{ying_overview_2019}). Additionally, over-fitting occurs when there is insufficient training data, or when there is a high hypothesis complexity, in terms of large weights, a large number of neurons, and extremely deep architectures. Therefore, designing neural network architectures that are less prone to over-fitting is a pertinent issue in supervised machine learning. 
\cite{sabiri2022mechanism} provides an overview of popular solutions to prevent over-fitting. For example, one of the most common solutions to over-fitting is regularization, such as $\ell_1$ or $\ell_2$ regularization (e.g. \cite{l1andl2multihinge,cortes2012l2}), which implements the size of the weights as a soft constraint. Other popular solutions include dropout (e.g. \cite{baldi2014dropout,srivastava2014dropout}), where certain neurons are dropped out during training time with a certain specified probability, and early stopping, where the training is stopped using a predefined predicate, usually when the validation error reaches a minimum (e.g. \cite{baldi2013understanding}). For example, in our previous work~\cite{wu2021machine1}, the Monte Carlo dropout technique was utilized in the development of NNs to mitigate the impact of data noise and reduce over-fitting. In addition to the above solutions, LCNNs have been demonstrated to be able to efficiently avoid over-fitting by constraining the Lipschitz constant of a network. However, at this stage, a fundamental understanding of the capability of LCNNs in reducing over-fitting in terms of the generalization ability of LCNNs over the underlying data distribution is missing.

Motivated by the above considerations, in this work, we incorporate LCNNs using SpectralDense layers in MPC and demonstrate that the LCNNs can effectively resolve the two aforementioned issues: sensitivity to input perturbations and over-fitting in the presence of noise. Rigorous theoretical results are developed to demonstrate that LCNNs are provably robust against input perturbations because of their low Lipschitz constant and provably robust against over-fitting due to their lowered hypothesis complexity, i.e., low Rademacher complexity. 
The rest of this article is organized as follows. In Section~\ref{sec:preliminaries}, the nonlinear systems that are considered and the application method of FNNs in MPC are first presented. In Section~\ref{sec:LCNN}, we present the formulation of LCNNs using SpectralDense layers, followed by a discussion on their improved robustness against input perturbations. In Section~\ref{sec:universal_app_thm}, we prove the universal approximation theorem for 1-Lipschitz continuous functions for LCNNs using SpectralDense layers. In Section~\ref{sec:generalization}, we develop a probabilistic generalization error bound for LCNNs to show that LCNNs using SpectralDense layers can effectively prevent over-fitting. This is done by computing an upper bound on the empirical Rademacher complexity method (ERC) of the function class represented by LCNNs using SpectralDense layers.  Finally, in Section~\ref{sec:application}, we carry out a simulation study of a benchmark chemical reactor example, where we will exhibit the superiority of LCNNs over conventional FNNs with dense layers in the presence of noisy training data.

\section{Preliminaries}\label{sec:preliminaries}

\subsection{Notations}
 $\lVert W \rVert_F$ and $\lVert W \rVert_2$ denote the Frobenius norm and the spectral norm of a matrix $W \in \mathbb{R}^{n \times m}$ respectively.  $\mathbb{R}^{\geq 0}$ denotes the set of all nonnegative real numbers. A function $f :\mathbb{R}^n \to \mathbb{R}^m $ is continuously differentiable if and only if it is differentiable and the Jacobian of $f$, denoted by $J_f$, is continuous. Given a vector $x \in \mathbb{R}^n$, let $\lVert x \rVert$ denote the Euclidean norm of $x$.
A metric space is a set $X$ equipped with a metric function $d_X :X \times X \to \mathbb{R}^{\geq 0 }$ such that 1) $d_X (x,y) = 0$ if and only if $x = y$, and 2) for all $x,y,z \in X$,  the triangle inequality $d_X (x,z) \leq d_X (x,y) + d_X (y,z)$ holds. We denote a metric space as an ordered pair $(X, d_X)$.  
Given an event $A$, we denote $\mathbb{P}(A)$ to be its probability. Given a random variable $X$, we denote $\mathbb{E}[X]$ to be its expectation.

\subsection{Class of Systems}\label{classofsystems}
The nonlinear systems that are considered in this article can be represented by the following ordinary differential equation (ODE):
\begin{equation}\label{eq:nonlin_sys}
\dot{x}=F(x,u):=f(x)+g(x)u
\end{equation}
Here $x \in \mathbb{R}^n$ is the current state vector,  $u \in \mathbb{R}^m$ is the control vector, and $f : \mathbb{R}^{n} \to \mathbb{R}^n$, $g : \mathbb{R}^n \to \mathbb{R}^{n \times m}$ are continuously differentiable matrix-valued functions. We also assume that $f(0) = 0$, so that the origin $(x,u) = (0,0)$ is an equilibrium point. 

We also assume that there is a Lyapunov function $ V : D \to \mathbb{R}^{\geq 0}$ that is continuously differentiable and is equipped with a controller $\Phi : D \to U $, such that the origin $(x,u) = (0,0)$ is an equilibrium point that is exponentially (closed-loop) stable. Here $D$ and $U$ are compact subsets of $\mathbb{R}^n$ and $\mathbb{R}^m$, respectively, that contain an open set surrounding the origin. In addition, the stability region for this controller $\Phi(x)$ is then taken to be a sublevel set of $V$, i.e., $\Omega_\rho := \{ x~ | \, V(x) \leq \rho \}$ with $\rho$ positive and $\Omega_\rho \subset D$. 
Additionally, since $f, g$ are continuously differentiable, the following inequalities can be readily derived for any $x, x' \in D$, $u \in U$ and some constants $K_F$, $L_x$:
\begin{subequations}
\begin{equation}
\lVert F(x,u)\rVert \leq K_F 
\end{equation}
\begin{equation}\label{eq:lipschitzinx}
\lVert F(x',u) - F(x,u)\rVert \leq L_x \lVert x'-x\rVert
\end{equation}
\end{subequations}

\subsection{Feedforward Neural Network (FNN)}
This subsection gives a short summary of the development of FNNs for the nonlinear systems represented by Eq. \ref{eq:nonlin_sys}. Specifically, since the FNN is developed as the prediction model for model predictive controllers, we consider the FNNs that are built to capture the nonlinear dynamics of Eq.~\ref{eq:nonlin_sys} , whereby the control input actions are applied using the sample-and-hold method. This means that given an initial state $x_0$, the control action $u$ applied is constant throughout the entire period of the sampling time $\Delta > 0$. Suppose that the system state is currently at $x(0) = x_0$. From Eq. \ref{eq:lipschitzinx} and the Picard-Lindel\"{o}f Theorem for ODEs, there exists a unique state trajectory $x(t)$ such that 
\begin{equation}\label{eq:x(t)equation}
x(t) =  x_0 + \int_0^t F(x(s),u) ds
\end{equation}
We can then define $\tilde{F} : \mathbb{R}^n \times \mathbb{R}^m \to \mathbb{R}^n$ by $\tilde{F}(x_0,u) = x(\Delta) $, which is the function that takes the current state and the control action as inputs and predicts the state at $\Delta$ time later. This function can then be approximated by an FNN, such as an LCNN, and the approximation is denoted by $\tilde{F}_{nn}$. In order to develop the neural networks, open-loop simulations of the ODE system represented by Eq. \ref{eq:nonlin_sys} using varying control actions will be conducted to generate the required dataset. Specifically, we perform a sweep of all possible initial states $x_0 \in \Omega_\rho$ and control actions $u \in U$, and use the forward Euler method with integration time step $h_c \ll \Delta$, to deduce the value of $\tilde{F}(x_0,u)$, that is, to deduce the state at $t=\Delta$. The training dataset consists of all such possible pairs $(x_0,u)$ (FNN inputs)   $\tilde{F}(x_0,u)$ (FNN outputs), where $\tilde{F}(x_0,u)$ is the actual future state. Therefore, the two functions $x_{t+1} := \tilde{F}(x_t, u_t)$ and $x_{t+1} := \tilde{F}_{nn}(x_t, u_t)$  represent define two distinct nonlinear discrete-time systems respectively, where $x_{t+1}$ denotes the state at $t+\Delta$.

Since $\tilde{F}_{nn}$ will be the function used in the MPC optimization algorithm, ensuring that $\tilde{F}_{nn}$ is an accurate approximation of $\tilde{F}$ is necessary so that the neural network captures the nonlinear dynamics well. In general, the FNN model $\tilde{F}_{nn}$ should be developed with sufficient training data and an appropriate architecture in terms of the number of neurons and layers in order to achieve the desired prediction accuracy on both training and test sets. However, in the presence of insufficient training data or noisy data, over-fitting might occur, leading to a model that performs poorly on the test data set and generalizes poorly. Also, the developed FNN model for $\tilde{F}_{nn}$ should not be overtly sensitive to input perturbations (that is, it should not have overtly large gradients) in order to ensure that $\tilde{F}_{nn}$ can generalize well to other data points outside the training dataset but within the desired domain $D \times U$.  
Therefore, to address the issues of over-fitting and sensitivity, we will develop LCNNs for the nonlinear system of Eq. \ref{eq:nonlin_sys} in this work, and show that LCNNs can overcome the limitations of conventional FNNs with dense layers by lowering sensitivity and preventing over-fitting.

\section{Lipschitz-Constrained Neural Network Models Using SpectralDense Layers}\label{sec:LCNN}
In this section, the architecture of LCNNs using SpectralDense layers will first be introduced, followed by a discussion on the reduced sensitivity of LCNN to input perturbations as compared to conventional FNNs. First, we begin with an important definition.
\begin{definition}
A function $f: X \to Y$  where $X \subset \mathbb{R}^n$  and $Y \subset \mathbb{R}^m$ is \textbf{Lipschitz continuous} with Lipschitz constant $L$  (or $L$-Lipschitz) if $\forall x,y \in X$, one has 
$$\lVert f(x)-f(y) \rVert \leq L \cdot \lVert x-y \rVert$$
\end{definition}
It is readily shown that if $f$ is $L$-Lipschitz continuous, given a small perturbation to the input, the output $f(x)$ changes by at most $L$ times the magnitude of that perturbation. As a result, as long as the Lipschitz constant of a neural network is constrained to be a small value, it is less sensitive with respect to input perturbations. In the next subsection, we demonstrate that LCNNs using SpectralDense layers have a constrained small Lipschitz constant, where each of the SpectralDense layers has a Lipschitz constant of 1.

\subsection{SpectralDense layers}
The mathematical definition of the SpectralDense layers used to construct an LCNN is first presented. Recall that by the singular value decomposition (SVD), for any $W \in \mathbb{R}^{m \times n}$, there exist orthogonal matrices $U \in \mathbb{R}^{m \times m}$, $V \in \mathbb{R}^{n \times n}$, and a rectangular diagonal matrix $D \in \mathbb{R}^{m \times n}$ with positive entries such that $W = UDV^{T}$. First, we recall the definition of a conventional dense layer:
\begin{definition}
A \textbf{dense layer} is a function $f : \mathbb{R}^n \to \mathbb{R}^m$ of the form
$$ f : x  \to  \sigma (Wx + b)$$
where $ \sigma \in \mathbb{R}^m \to \mathbb{R}^m$ is an activation function,
$W \in \mathbb{R}^{ m \times n} $ is a weight matrix, and $b \in \mathbb{R}^n$ is a bias term. A dense layer is a layer that is deeply connected with its preceding layer (i.e., the neurons of the layer are connected to every neuron of its preceding layer).
\end{definition}
It should be noted that in conventional dense layers, the activation function $ \sigma$ is applied component-wise, that is,  $ \sigma(x_1, x_2 , \dots , x_m) = (\sigma'(x_1), \sigma'(x_2), \dots , \sigma'(x_m) ) $ where $ \sigma' : \mathbb{R} \to \mathbb{R}$  is a real-valued function, such as ReLU or $\tanh$. However, in SpectralDense layers, the following GroupSort function is used as the activation function $\sigma$:
\begin{definition} (\cite{anil2019sorting})
The \textbf{GroupSort function} (of group size 2) is a function $\sigma: \mathbb{R}^m \to \mathbb{R}^m $ defined as follows: 
\begin{subequations}
If $m$ is even, then 
\begin{equation}
\sigma([x_1,x_2,\cdots,x_{m-1},x_m]^T) = [\max(x_1,x_2),\min(x_1,x_2),\cdots,\max(x_{m-1},x_m),\min(x_{m-1},x_m)]^T
\end{equation}
else, if $m$ is odd, 
\begin{equation}
\sigma([x_1,x_2,\cdots,x_{m-2},x_{m-1},x_m]^T) = [\max(x_1,x_2),\min(x_1,x_2),\cdots,\min(x_{m-2},x_{m-1}),x_m]^T
\end{equation}

\noindent For example, in the case where the output layer has dimension $m =4$, we have
\begin{equation}
\sigma ([0,3,4,2]^T) =  [3,0,4,2]^T \;\;\; \sigma([5,3,2,4]^T) =  [5,3,4,2]^T
\end{equation}
\end{subequations}
\end{definition}
\noindent  SpectralDense layers can now be defined as follows:
\begin{definition} (\cite{serrurier2021achieving})
\textbf{SpectralDense layers} are dense layers such that 1) the largest singular value of $W$ is 1, and 2) the activation function $\sigma$ is GroupSort function.
\end{definition}
Therefore, SpectralDense layers are similar to dense layers in terms of their structure, except that the activation function does not act component-wise and the weight matrices have a spectral norm of 1. The spectral norm of a matrix is also equal to the largest singular value in its SVD. Since the largest singular value of the weight matrix $W$ is 1, the spectral norm $\lVert W \rVert_2 $ is also 1. The function $\sigma : \mathbb{R}^m \to \mathbb{R}^m$ is also 1-Lipschitz continuous (with respect to the Euclidean norm), since it has Jacobian of spectral norm 1 almost everywhere (everywhere except a set of measure 0); thus it is also 1-Lipschitz continuous (see Theorem 3.1.6 in  \cite{federer2014geometric}). We therefore conclude that every SpectralDense layer is also 1-Lipschitz continuous.  Next,  the definition of the class of LCNNs is given as follows. 
 \begin{definition}
 	Let $\mathcal{LN}^m_n$ be the class of Lipschitz-constrained neural networks (LCNNs) as follows:
 \begin{equation}\label{LCNNdefinition}
 \begin{array}{l}
\mathcal{LN}_n^m := \{ \: f ~|~f : \mathbb{R}^n \to \mathbb{R}^m \:, 
\:   \exists j\in \mathbb{N} \text { such that } f = W_{j+1} f_j \circ f_{j-1} \circ ... \circ f_2 \circ f_1, \\
\text{ where } f_i = \sigma( W_i x + b ), 
\text{ and } \, \lVert W_i \rVert_2 = 1, i=1,...,j \, \}
\end{array}
\end{equation}
where $\sigma$ is the GroupSort activation with group size 2. Thus, each LCNN in $\mathcal{LN}_n^m$ consists of many SpectralDense layers (i.e., $W_{i}$, $i=1,...,j$) composed together, with a final weight matrix $W_{j+1}$ at the end. The spectral norm constraint is imposed for all weight matrices except the final weight matrix $W_{j+1}$.
\end{definition}
 Since the Lipschitz constant for each of the functions $f_i$, $i=1,...,j,$ is bounded by 1, it is readily shown that for each neural network in $\mathcal{LN}_n^m$, the Lipschitz constant is bounded by the spectral norm of the final weight matrix $W_{j+1}$. This implies that we can control the Lipschitz constant of an LCNN by manipulating the spectral norm of the final weight matrix $W_{j+1}$, and this can be done in a variety of ways, such as imposing the constraint for the final weight matrix as a whole, or restricting the absolute value of each entry in the final weight matrix during the training process.  In the simulation study in Section 5, we impose the absolute value constraint on each entry of the final weight matrix to control the Lipschitz constant of the LCNN.
 
\begin{rem}
The SpectralDense layers adopted in this work differ from those used in \cite{anil2019sorting}, since the matrix norm used in this work is the spectral norm, whereas the norm used by Anil et al. is the $\infty$-norm. Although all matrix norms give rise to equivalent topologies (the same open sets), we use the spectral norm as it is directly related to the Jacobian of the function. Specifically, a well-known theorem by Rademacher (see Lemma 3.1.7 of \cite{federer2014geometric} for a proof) states that if $X \subset \mathbb{R}^n$ is open and $f : X \to \mathbb{R}^m$ is $L$-Lipschitz continuous, then $f$ is almost everywhere differentiable, and we have 
\begin{equation}
 L =  \sup_{x \in X} \, \lVert J_f (x) \rVert_2
\end{equation}
Therefore, it is observed that the Lipschitz constant using the Euclidean norm on the input space $\mathbb{R}^n$ and output space $\mathbb{R}^m$ is actually the supremum of the spectral norm of the Jacobian matrix of $f$. If the $\infty$-norm were to be used, to the best of our knowledge, no such essential relationship that involves the Lipschitz constant has been proven to this date.
\end{rem}
\subsection{Robustness of  LCNNs}
We now discuss how LCNNs using SpectralDense layers can resolve the issue of sensitivity to input perturbations. Let $f: X \to \mathbb{R}^m$ be a neural network that has been trained using a training algorithm and $X \subset \mathbb{R}^n$ be an open subset such that $f$ is almost everywhere differentiable with Jacobian $J_f$. Given a set of training data, at each point $x \in X$, one plausible way to maximize the impact of the input perturbation is to traverse along the direction corresponding to the largest eigenvalue (the spectral norm) of $J_f (x)$ in its SVD, since this is the direction that leads to the largest variation in output (\cite{szegedy2013intriguing,goodfellow2014explaining}). 

For any $f \in \mathcal{LN}^m_n$, Eq.~\ref{LCNNdefinition} shows that the Lipschitz constant of $f$ is bounded by the spectral norm of the final weight matrix. If the spectral norm of the final matrix is small, the corresponding LCNN will have a small Lipschitz constant, making it difficult to perturb, even if we travel along the direction corresponding to the largest eigenvalue of $J_f$. Specifically, if the input perturbation is of size $\delta$, the output change is at most $L\times \delta$, where $L$ is constrained to be a   small value. Therefore,  one plausible method to reduce the sensitivity of neural networks to input perturbations is to constrain the Lipschitz constant of the networks.

However, since the Lipschitz constant affects the network capacity, controlling the upper bound of the Lipschitz constant in LCNNs could result in a reduced network capacity. To address this issue, we will demonstrate in the next section that the function class $\mathcal{LN}^m_n$ is a universal approximator for any Lipschitz continuous target function. Additionally, a pertinent question that arises is whether, in practice, the Lipschitz constants of conventional FNNs (e.g.,  FNNs using conventional dense layers and ReLU activation functions) are indeed much larger than those in LCNNs. In the special case of FNNs with ReLU activation functions, \cite{bhowmick2021lipbab} have designed a provably correct approximation algorithm known as Lipschitz Branch and Bound (LipBaB), which obtains the Lipschitz constant of such networks on a compact rectangular domain. In Section \ref{lipschitz constant comparison}, we will demonstrate empirically that with noisy training data, FNNs with dense layers could have a Lipschitz constant several orders of magnitude higher than that of LCNNs.

\section{Universal Approximation Theorem for LCNNs}\label{sec:universal_app_thm}
This section develops the universal approximation theorem for LCNNs using SpectralDense layers, which demonstrates that the LCNNs with a bounded Lipschitz constant can approximate any nonlinear function as long as the target function is Lipschitz continuous. Before we present the proof for vector-valued LCNNs that are developed for nonlinear systems with vector-valued outputs such that of Eq.~\ref{eq:nonlin_sys}, we first develop a theorem that considers the approximation of real-valued functions. Then, the results for real-valued functions can be generalized to the multi-dimensional output case. We first define the real-valued function class as follows.
\begin{equation}\label{LCNNdefinition2}
\mathcal{LN}_n := \{ \: f ~|~f : \mathbb{R}^n \to \mathbb{R} \:,\:   \exists j\in \mathbb{N}, s.t.~f = f_j \circ ... f_2 \circ f_1 , f_i = \sigma( W_i x + b ), \lVert W_i \rVert_2 = 1, i=1,...,j \}
\end{equation}
where $\sigma$ is the GroupSort activation with group size 2. The definition of Eq.~\ref{LCNNdefinition2} is similar to that of Eq. \ref{LCNNdefinition}, except that the functions are real-valued, and the spectral norm for each weight matrix is 1, including the final weight matrix. Note that the final map $f_j$ is an affine map, without any sorting, since the output of $f_j$ is a single real number. It readily follows that any function in $\mathcal{LN}_n$ is also 1-Lipschitz continuous since the spectral norm of the final weight matrix is one.

Given a target function $F : D \to \mathbb{R}$ where $D$ is a compact and connected domain and $F$ is Lipschitz continuous, we prove that real-valued functions from $\mathcal{LN}_n$ are universal approximators of 1-Lipschitz functions, provided that we allow for an amplification of at most $\sqrt{2}$ at the end. In principle, this implies that LCNNs can 
approximate any Lipschitz function (i.e., they are dense with respect to the uniform norm on a compact set). We will follow the notation in \cite{anil2019sorting}, but with slight modifications. We first present the following definitions, which will be used in the proof of the universal approximation theorem.

\begin{definition}
Let $(X,d_X)$ be a metric space. We use $C_L (X,\mathbb{R})$ to denote the set of all 1-Lipschitz real-valued functions on $X$.
\end{definition}
\begin{definition}
Let $A$ be a set of functions from $\mathbb{R}^n $ to $\mathbb{R}$, and let $k$ be a real number. We define $kA$ as follows.
\begin{equation}
kA := \{\;cf \; | \; |c| \leq k,\; f \in A \}
\end{equation}
\end{definition}
\begin{definition}\label{def:Lattice}
A \textbf{lattice} $\mathcal{L}$ in $C_L (X,\mathbb{R})$ is a set of functions that is closed under point-wise maximums and minimums, that is, $\forall f,g \in \mathcal{L}$, $\min(f,g) , \max(f,g)\in \mathcal{L}$.
\end{definition}
The following restricted Stone-Weierstrass theorem allows us to approximate 1-Lipschitz continuous functions using lattices.
\begin{thm}[Restricted Stone-Weierstrass~\cite{anil2019sorting}]\label{th:RSWThm}
Let $(X, d_X)$ be a compact metric space and $\mathcal{L}$ be a lattice in $C_L (X,\mathbb{R})$. Suppose that for all $a,b \in \mathbb{R}$ and $x,y \in X$ such that $|a-b| \leq d_X (x,y)$, there exists an $f \in \mathcal{L} $ such that $f(x) = a $ and $f(y) = b$. Then $\mathcal{L}$ is dense in $C_L (X,\mathbb{R})$ with respect to the uniform topology, that is, for every $\epsilon >0$ and  for every $f \in C_L (X,\mathbb{R})$, there exists an $\tilde{f} \in \mathcal{L}$ such that
\begin{equation}
\sup_{x \in X} \, | f(x) - \tilde{f}(x) | < \epsilon
\end{equation} 
\end{thm}
The proof of Theorem \ref{th:RSWThm} is given in  \cite{anil2019sorting} and is omitted here. Based on Theorem 1, we develop the following theorem to prove that the LCNN networks constructed using SpectralDense layers can also serve as universal approximators, if we allow for an amplification of the output at the end. The proof uses Theorem \ref{th:RSWThm} to show that $\sqrt{2}\mathcal{LN}_n \cap C_L (D,\mathbb{R})$ is a lattice. The proof techniques and structure are similar to those in \cite{anil2019sorting}, while the key difference is the use of the SVDs of the weight matrices since we are using the spectral norm instead.

\begin{thm}\label{universalapproximationtheorem}
Let $D \subset \mathbb{R}^n$ be a compact subset, and $\mathcal{LN}_n$ be the set of LCNNs defined in Eq. \ref{LCNNdefinition}.   $\sqrt{2}\mathcal{LN}_n \cap C_L (D,\mathbb{R})$ is dense in $C_L (D,\mathbb{R})$ with respect to the uniform topology, i.e., for every $\epsilon >0$ and for every $f \in C_L (D,\mathbb{R})$, there exists an $\tilde{f} \in \sqrt{2}\mathcal{LN}_n \cap C_L (D,\mathbb{R})$ such that
\begin{equation}\label{eq:thm_universal:result}
	 \sup_{x \in D} \, \lvert f(x) - \tilde{f}(x) \rvert < \epsilon 
\end{equation} 
\end{thm}
 \begin{proof}
 	 Theorem~\ref{universalapproximationtheorem} states that if we allow for an amplification of $\sqrt{2}$ at the end, then  the LCNNs using SpectralDense layers defined in Eq. \ref{LCNNdefinition} can approximate 1-Lipschitz continuous functions arbitrarily accurately. To prove  Eq.~\ref{eq:thm_universal:result},  we first show that $\mathcal{L} := \sqrt{2}\mathcal{LN}_n\cap C_L (D,\mathbb{R})$ satisfies the assumptions needed for the restricted Stone-Weierstrass theorem. 

Note that for all $a,b \in \mathbb{R}$ and $x,y \in D$ such that $|a-b| \leq \lVert x - y\rVert$, there exists an $f \in \mathcal{L} $ such that $f(x) = a $ and $f(y) = b$. To construct such an $f$, let $f(v) = w^T (v - x) + a$ and choose $w \in \mathbb{R}^{n} $ with $\lVert w \rVert = 1$ carefully so that this holds. We need to ensure that $w^T (y- x) = b-a$ such that $f(y) = b$. Therefore, we need to choose $w$ such that the following equation holds.
\begin{equation}
    b-a = w^T (y- x) \leq \lVert w \rVert \cdot \lVert y-x  \rVert =  \lVert y-x  \rVert
\end{equation}
We can first set $w$ to be in the same direction as $y-x$, and then gradually rotate $w$ away from $y-x$ (for example, using a suitable orthogonal matrix) so that we eventually have $b-a = w^T (y- x)$.  
\par
Next, we need to show that $\mathcal{L}$ is a lattice. This is equivalent to showing that $\mathcal{L}$ is closed under pointwise maximums and minimums by Definition \ref{def:Lattice}. Following the same proof as in \cite{anil2019sorting}, we assume that $f,g \in \mathcal{L}$ are defined by their weights and biases:
\begin{equation}
[W_1^f, b_1^f , W_2^f ,b_2^f , \dots , W_{d_f}^f , b_{d_f}^f ] \;\;\;  [W_1^g, b_1^g , W_2^g ,b_2^g , \dots,  W_{d_g}^g , b_{d_g}^g ]
\end{equation}
\noindent where $d_f$ and $d_g$ represent the depths of the networks $f$ and $g$, respectively. We assume without loss of generality that the two neural networks $f$ and $g$ have the same depth, i.e.,  $d_f =d_g$. This is possible since if $d_f > d_g$, one can pad the neural network with identity matrix weights and zero biases until they are of the same depth, and likewise for $d_f < d_g$. 
We also assume without loss of generality that each of the weight matrices (except the final weight matrix) has an even number of rows. In the case where the neural network has a weight matrix with an odd number of rows, the weight matrix can be padded with a zero row under the last row and with a bias $-M$ where $M>0$ is sufficiently large in that row (this is possible since $D$ is compact and by the extreme value theorem). This is to prevent a different sorting configuration of the output vector of that layer. Then, in the next matrix, we add a column of zeros to remove the $-M$ entry.

We now construct a neural network $h$ for $ \max(f,g)$ and $\min(f,g)$ with new weights such that each of the weights satisfies $\lVert W_i^h \rVert_2 = 1$ for $i = 1, \cdots, d_f $, and the scaling factor of $\sqrt{2}$ will be applied later. 
We first construct a suitable neural network with the layers of $f$ and $g$ side by side, but with some modifications.   Specifically, the first matrix $W_1^h$ and the first bias $b_1^h$ are designed as follows:
\begin{equation} \label{eq:W1h}
W_1^h = c
[
W_1^f ,
W_1^g 
]^T
\,\,\,\,
    b_1^h = 
[b_1^f ,
b_1^g ]
\end{equation}
where $c$ is a positive constant chosen based on $W_1^f $ and $W_1^g$ to ensure that $\lVert W_1^h \rVert_2 = 1$.
From Eq.~\ref{eq:W1h}, it is shown that the output of the first layer is obtained by concatenating the output of the first layer of $f$ and $g$  together, and multiplied by a positive constant $c$. Note that $1 \leq c \leq \frac{1}{\sqrt{2}}$ since each of $W_1^f$ and $W_1^g$ have spectral norms 1. Then, for the rest of the layers ($i \geq 2$), we define 
\begin{equation}
W_i^h =
\begin{bmatrix}
W_i^f & 0\\
0 & W_i^g 
\end{bmatrix}
\end{equation}
To prove that $\lVert W_i^h \rVert_2 = 1$, we use the singular value decomposition. 
\begin{subequations}
\begin{align}
W_i^h &=
\begin{bmatrix}
W_i^f & 0\\
0 & W_i^g 
\end{bmatrix}
= \begin{bmatrix}
U_i^f D_i^f  {V_i^f}^*   & 0\\
0 & U_i^g D_i^g {V_i^g}^* 
\end{bmatrix} \\ 
&= 
\begin{bmatrix}
U_i^f & 0\\
0 & U_i^g 
\end{bmatrix} 
\begin{bmatrix}
 D_i^f  & 0\\
0 & D_i^g 
\end{bmatrix} 
\begin{bmatrix}
 {V_i^f}^*   & 0\\
0 &  {V_i^g}^* 
\end{bmatrix} \\ 
&= UDV^T \label{eqn:SVD}
\end{align}
\end{subequations}
It is noted that $\lVert W_i^h \rVert_2$ is simply the largest singular value of its singular value decomposition. On the right-hand side (RHS) of Eq. \ref{eqn:SVD}, even if the matrix $D$ is not rectangular block diagonal, we can permute the columns of $D$ and the rows of $V^T$ simultaneously, and then the rows of $D$ and columns of $U$ simultaneously, to obtain a new rectangular block diagonal matrix. Permuting the columns of $U$ and the rows of $V^T$ does not change the unitary property of these matrices. Therefore, the largest singular value of $W_i^h$ is 1 since both $D_i^f$ and $D_i^g$ have the largest singular values of 1, and therefore $\lVert W_i^h \rVert_2 = 1$. We also take the following. 
\begin{equation}
b_i^h = 
c
[b_i^f,
b_i^g] 
\end{equation}
for each of the biases. Finally, after passing through the GroupSort activation function, the output of the last layer is  
\begin{equation}\label{eq:beforefinaloutput}
c [ \max(f(x), g(x)), \min(f(x),g(x))]^T
\end{equation}
By passing Eq. \ref{eq:beforefinaloutput} through the weight matrix $[1,0]$ or $[0,1]$, we obtain $h(x) = c\max(f(x), g(x))$ or $h(x) = c\min(f(x), g(x))$ respectively. Since we consider the set $\mathcal{L} = \sqrt{2}\mathcal{LN}_n\cap C_L (D,\mathbb{R})$, and $1\leq c^{-1} \leq \sqrt{2}$, we observe that $c^{-1}  \, h(x) = \max(f(x), g(x))$ or $c^{-1} \, h(x) = \min(f(x), g(x))$. This proves that $\max(f(x), g(x))$ and $\min(f(x), g(x))$ are in $\sqrt{2}\mathcal{LN}_n \cap C_L (D,\mathbb{R})$, which implies that $\mathcal{L}$ is a lattice. Therefore, $\sqrt{2}\mathcal{LN}_n \cap C_L (D,\mathbb{R})$ satisfies the assumptions needed for the restricted Stone-Weierstrass theorem in Theorem \ref{th:RSWThm}, and this completes the proof.
 \end{proof}
 Theorem~\ref{universalapproximationtheorem} implies that for 1-Lipschitz target functions, if we allow for an amplification of $\sqrt{2}$ at the last layer, then LCNNs using SpectralDense layers are universal approximators for the target function. Similarly, LCNNs can approximate an $L$-Lipschitz continuous function so long as an amplification of $\sqrt{2}L$ is allowed. The final amplification can be easily implemented by using a suitable weight matrix such as a constant multiple of the identity matrix as the final layer.
 \begin{rem}
 The above theorem can be generalized to regression problems from $\mathbb{R}^n$ to $\mathbb{R}^m$. Specifically, a similar result can be derived for LCNNs developed for $f : \mathbb{R}^n \to \mathbb{R}^m$ by approximating each component function $f_i$ for $ i=1,...,m$ with a 1-Lipschitz neural network, and then padding the neural networks together to form a large neural network with $m$ outputs. However, note that, as a result, the resulting approximating function might be $\sqrt{m}$-Lipschitz continuous.
 \end{rem}
\section{Preventing Over-fitting and Improving Generalization with LCNNs}\label{sec:generalization}
In this section, we provide a theoretical argument using the empirical Rademacher complexity to show that LCNNs can prevent over-fitting and therefore generalize better than conventional FNNs by computing the empirical Rademacher complexity of the LCNNs, and comparing that with usual FNNs made using the same architecture (i.e., same number of neurons in each layer and same depth). Specifically, we develop a bound of the empirical Rademacher complexity (ERC) for any FNN that utilizes the GroupSort function (hereby called GroupSort Neural Networks), and subsequently use this to obtain a bound for LCNNs using SpectralDense layers as an immediate corollary. Then, we compare this bound with the bound for the empirical Rademacher complexity of FNNs using 1-Lipschitz component-wise activation functions (e.g. ReLU) and show that the GroupSort NNs achieve a tighter generalization error bound. 
\subsection{Assumptions and Preliminaries}\label{Assumptions}
We first assume that the input domain is a bounded subset $\mathcal{X} \subset \mathbb{R}^{d_x}$, where for each $ x \in \mathcal{X}$, one has $ \lVert x \rVert \leq B$, and the output space is a subset of the vector space $\mathcal{Y} \subset \mathbb{R}^{d_y}$. We also assume that each input-output pair in the dataset $(x,y)$ is drawn from some distribution $ \mathcal{D}  \subset\mathcal{X}  \, \times \, \mathcal{Y} $ with some probability distribution $\mathbb{P}$. We assume that $L : \mathcal{Y} \times \mathcal{Y} \to \mathbb{R}^{\geq 0}$ is a loss function that satisfies the following two properties: 1) there exists a positive $M > 0$  such that for all $y,y' \in \mathbb{R}$ $L(y,y') \leq M$, and 2) for all $y' \in \mathbb{R}^{d_y}$, the function $y \to  L(y,y')$ is $L_r$-Lipschitz for some $L_r > 0$.

\par
Let $h : \mathcal{X} \to  \mathcal{Y}$ be a function that represents a neural network model (termed hypothesis) in a hypothesis class. 
\begin{definition}
We define \textbf{generalization error} as
\begin{equation}
\underset{(x,y) \sim \mathcal{D}}{\mathbb{E}} \, [L(h(x),y)] =  \int_{ \mathcal{X} \times \mathcal{Y} } L(h(x),y) \, \mathbb{P} (dx \times dy)
\end{equation}
\end{definition}
\begin{definition}
For any dataset $S = \{s_1, s_2, \cdots , s_m \} \subset \mathcal{X} \times \mathcal{Y}$ with $s_i = (x_i, y_i)$, we define the \textbf{empirical error} as 
\begin{equation}
\hat{\mathbb{E}}_S [L(h(x),y)]  = \frac{1}{m} \sum_{i=1}^m L(h(x_i),y_i)
\end{equation}
\end{definition}
The generalization error is a measure of how well a hypothesis $h : \mathcal{X} \to \mathcal{Y}$ generalizes from the training dataset to the entire domain being considered. On the other hand, the empirical error is a measure of how well the hypothesis $h$ performs on the available data points $S$.

\begin{rem}
In this work, we use the $\ell_2$ error for the loss function $L$ in the NN training process, i.e., $L(y,y') = \lVert y - y' \rVert^2 $, which is locally Lipschitz continuous but not globally Lipschitz continuous. However, since we consider an input domain of $D \times U$ which is a compact set, and the function $\tilde{F}$ is also continuous, the range of $\tilde{F}$ is also compact and bounded. As a result, the output domain $\mathcal{Y} \subset \mathbb{R}^{d_y}$ is   bounded,  
and the  assumptions for the loss function $L(\cdot,\cdot)$ are  satisfied using the $\ell_2$ error.
\end{rem}

\subsection{Empirical Rademacher Complexity Bound for  GroupSort Neural Networks}
We first define the empirical Rademacher complexity of a real-valued function hypothesis class.
\begin{definition}\label{ERCdefinition}
Given an input domain space $\mathcal{X} \subset \mathbb{R}^{d_x}$, suppose $H$ is a class of real-valued functions from $\mathcal{X}$ to $\mathbb{R}$. Let $S  = \{ x_1, x_2, \cdots , x_m \} \subset \mathcal{X}$, which is a set of samples from $\mathcal{X}$. The \textbf{empircial Rademacher complexity} (ERC) of $S$ with respect to $H$, denoted by $\mathcal{R}_S (H)$, is defined as:
\begin{equation}
\mathcal{R}_S (H) := \underset{\epsilon}{\mathbb{E}}  \, \sup_{h \in H } \frac{1}{m} \sum_{i = 1}^m  \epsilon_i h(x_i)
\end{equation}
where each of $\epsilon_i$ are i.i.d Rademacher variables, i.e., with $\mathbb{P}(\epsilon_i = 1) = \frac{1}{2}$ and $\mathbb{P}(\epsilon_i = -1) = \frac{1}{2}$.
\end{definition}
The ERC measures the richness of a real-valued function hypothesis class with respect to a probability distribution. A more complex hypothesis class with a larger ERC is likely to represent a richer variety of functions, but may also lead to over-fitting, especially in the presence of noise or insufficient data. The ERC is often used to obtain a probabilistic bound on the generalization error in statistical machine learning. Specifically, we first present the following theorem in \cite{mohri2018foundations} that obtains a probabilistic upper bound on the generalization error for a hypothesis class.

\begin{thm}[Theorem 3.3 in \cite{mohri2018foundations}]\label{mohri}
Let $\mathcal{H}$ be a hypothesis class of functions  $h: \mathcal{X} \subset \mathbb{R}^{d_x} \to \mathcal{Y} \subset \mathbb{R}^{d_y}$ and $L : \mathcal{Y} \times \mathcal{Y} \to \mathbb{R}^{\geq 0} $ be the loss function that satisfies the properties in Section \ref{Assumptions}. Let $\mathcal{G}$ be a hypothesis class of loss functions associated with the hypotheses $h\in \mathcal{H}$:
\begin{equation}\label{eq:hypothesisclassoflossfunctions}
\mathcal{G} := \{ g : (x,y) \to L(h(x),y) \, , \, h \in \mathcal{H} \}
\end{equation}
For any set of $m$ i.i.d. training samples $S = \{s_1, s_2, \cdots , s_m \}$, $s_i = (x_i,y_i) $, $i=1,...,m$, drawn from a probability distribution $\mathcal{D} \subset \mathcal{X} \times \mathcal{Y} $,  for any $\delta \in (0,1)$, the following upper bound holds with probability $1-\delta$ :
\begin{equation}\label{eq:thm3B}
\underset{(x,y) \sim \mathcal{D}}{E} \, [L(h(x),y)] 
\leq \frac{1}{m} \sum_{i=1}^m L(h(x_i),y_i) + 2 \mathcal{R}_S (\mathcal{G}) + 3M \sqrt{\frac{\log \frac{1}{\delta}}{2m}} 
\end{equation}
\end{thm}

   Eq.~\ref{eq:hypothesisclassoflossfunctions} in Theorem~\ref{mohri} explains why an overtly large hypothesis class with a high degree of complexity could inevitably increase the generalization error. Specifically, if the hypothesis class is overtly large, the first term in the RHS of Eq.~\ref{eq:hypothesisclassoflossfunctions} (i.e., the empirical error $\hat{E}_S [L(h(x),y)] = \sum_i L(h(x_i),y_i)$) is expected to be sufficiently small, since one could find a hypothesis $h$ that minimizes the training error using an appropriate training algorithm. However, the trade-off is that the ERC $\mathcal{R}_S (\mathcal{G})$ will increase with the size of the hypothesis class. On the contrary, if the hypothesis class is limited,  the model complexity represented by the ERC of $\mathcal{R}_S (\mathcal{G})$ is reduced, while the empirical error $\hat{E}_S [L(h(x),y)]$ is expected to increase.   
Since the ERC of the hypothesis class of loss functions $\mathcal{R}_S (\mathcal{G})$ appears on the RHS of the inequality above, it implies that a lower ERC leads to a tighter generalization error bound. Therefore, in this section, we demonstrate that LCNNs using SpectralDense layers have a smaller ERC $\mathcal{R}_S (\mathcal{G})$ than conventional FNNs using dense layers.
\par
Since $\mathcal{Y} \subset \mathbb{R}^{d_y}$ is a subset of a real vector space, while the definition of ERC in Definition \ref{ERCdefinition} is with respect to a real-valued class of functions, to simplify the discussion, we can first apply the following contraction inequality. 
\begin{lemma}[Vector Contraction Inequality ~\cite{maurer2016vector}]\label{lemma:VCI} 
Let $\mathcal{H}$ be a hypothesis class of functions from $\mathcal{X} \subset \mathbb{R}^{d_x}$ to $\mathcal{Y} \subset \mathbb{R}^{d_y}$. For any set of data points $S = \{s_1, s_2, \cdots , s_m \} \subset \mathcal{X} \times \mathcal{Y}$ and $L_r$-Lipschitz  loss function $y \to  L(y,y')$   for some $L_r > 0$,  we have 
\begin{equation}\label{VCI}
\mathcal{R}_S (\mathcal{G}) = \underset{\epsilon}{\mathbb{E}}  \, \sup_{h \in \mathcal{H} } \frac{1}{m} \sum_{i = 1}^m  \epsilon_i L( h(x_i) , y)  \leq \sqrt{2} L_r  \underset{\epsilon}{\mathbb{E}}  \, \sup_{h \in \mathcal{H} } \frac{1}{m} \sum_{i = 1}^m   \sum_{k=1}^{d_y} \epsilon_{ik} h_k(x_i)
\end{equation}
where $h_k$ is the $k^{\text{th}}$ component function of $h \in \mathcal{H}$ and each of $\epsilon_{ik}$ are i.i.d. Rademacher variables.
\end{lemma}
The proof of the above inequality can be found in \cite{maurer2016vector}, and is omitted here. The RHS of the inequality in Eq.~\ref{VCI} can be bounded by taking out the sum in the supremum:
\begin{equation}\label{VCI2}
\underset{\epsilon}{\mathbb{E}}  \sup_{h \in \mathcal{H} } \frac{1}{m} \sum_{i = 1}^m   \sum_{k=1}^{d_y} \epsilon_{ik} h_k(x_i) \leq  \sum_{k=1}^{d_y} \underset{\epsilon}{\mathbb{E}}  \sup_{h_k \in \mathcal{H}_k } \frac{1}{m} \sum_{i = 1}^m    \epsilon_{ik} h_k(x_i) 
\end{equation}
where each of the $\mathcal{H}_k$ is a real-valued function class that correspond to the $k^{\text{th}}$ component function of each function in $h \in \mathcal{H}$. The RHS of Eq.~\ref{VCI2} is the sum of the Rademacher complexities of the hypothesis classes $\mathcal{H}_k$. Therefore, to obtain an upper bound for $\mathcal{R}_S (\mathcal{G})$, we can first consider the case of a real-valued function hypothesis class, and then extend the results to the multidimensional case by applying Eq.~\ref{VCI2}. Specifically, we first develop an ERC bound for GroupSort Neural Networks, which refers to any FNN that utilizes the GroupSort activation function. Since LCNNs  using SpectralDense layers are FNNs that also utilize the GroupSort activation function, the ERC bound for LCNNs using SpectralDense layers will follow as an immediate corollary.   

The following definitions are first presented to define the classes of real-valued GroupSort neural networks and of real-valued LCNNs using SpectralDense layers.
\begin{definition}\label{GNN}
 We use $\mathcal{H}_d$ to denote the hypothesis class of GroupSort neural networks with depth $d$ that map the input domain $\mathcal{X} \subset \mathbb{R}^{d_x}$ to $\mathbb{R}$: 
\begin{equation}\label{eq:H_d}
x \to W_d  \, \sigma  (W_{d-1}  \, \sigma ( \cdots \sigma (W_1 x)))
\end{equation}
where each of the weight matrices $W_i \in \mathbb{R}^{m_i \times n_i}$ has a bounded Frobenius norm, that is, $\lVert W_i\rVert_F \leq R_i$ for some $R_i \geq 0$, $m_i$ and $n_i$ are even for all $i=1,...,d-1$ except for $n_1$, and $\sigma$ is the GroupSort function with group size 2.
\end{definition}

\begin{definition}\label{SNN}
 We use $\mathcal{H}^{SD}_d$ to denote the hypothesis class of LCNNs using SpectralDense layers with depth $d$ that map the input domain $\mathcal{X} \subset \mathbb{R}^{d_x}$ to $\mathbb{R}$:
\begin{equation} \label{eq:H_SD}
x \to W_d  \, \sigma  (W_{d-1}  \, \sigma ( \cdots \sigma (W_1 x)))
\end{equation}
where each of the weight matrices satisfies $\lVert W_i\rVert_2 = 1$, $m_i$ and $n_i$ are even for all $i = 1,...,d-1$ except for $n_1$, and $\sigma$ is the GroupSort function with group size 2. 
\end{definition}

Note that the key difference between $\mathcal{H}_d$ and  $\mathcal{H}^{SD}_d$ is that 
$\mathcal{H}_d$ is defined for any FNN using GroupSort activation functions, while $\mathcal{H}^{SD}_d$ is defined for LCNNs that use the GroupSort activation function and satisfy $\lVert W_i\rVert_2 = 1$. Additionally, the Frobenius norm is used in $\mathcal{H}_d$, while the spectral norm is used in $\mathcal{H}^{SD}_d$. 
 Despite the differences between $\mathcal{H}_d$ and  $\mathcal{H}^{SD}_d$, it will be demonstrated from the following lemmas and theorems that the two classes $\mathcal{H}_d$ and  $\mathcal{H}^{SD}_d$ are highly related.

\begin{rem}
The bias term is ignored in Eq.~\ref{eq:H_d} and Eq.~\ref{eq:H_SD} to simplify the formulation, as in principle, we can take into account the bias term by padding the input with a vector consisting of ones, and introducing another weight matrix into each of the $W_i$.  Additionally,
we assume that $m_i$ and $n_i$ are even without any loss of generality, since this does not affect the expressiveness of the class of networks using LCNNs using SpectralDense layers, as shown in the proof of Theorem \ref{universalapproximationtheorem}.
\end{rem}

Next, we develop a bound on the ERC of $\mathcal{H}_d$, i.e., $ \mathcal{R}_S (\mathcal{H}_d)$, and then the ERC bound for $\mathcal{H}^{SD}_d$ follows as an immediate corollary. The main intuition to obtain such an upper bound for $\mathcal{R}_S (\mathcal{H}_d)$ is to recursively ``peel off'' the weight matrices and activation functions. Such methods were used in \cite{golowich2018size,neyshabur2015norm,wu2021statistical,wu2022statistical}, where the activation function was applied element-wise. The key difficulty in our setting stems from the fact that the GroupSort activation function $\sigma$ is not an element-wise function. To address this issue, we first represent the functions $\max(a,b)$ and $\min(a,b)$ as follows:

\begin{equation}\label{eq:max-min}
\max(a,b) = \frac{1}{2} (a + b + |b-a| )
\,\,\,\,
\min(a,b) = \frac{1}{2} (a + b - |b-a| )
\end{equation}

Before we present the results for peeling off the weight matrices of FNNs, the following definition is first given and will be used in the proof of Lemma~\ref{peelingoff} that  peels off one GroupSort activation function layer.
\begin{definition}\label{def:vectorvalued}
     For $d \geq 1$, we define $\tilde{\mathcal{H}}_d$ as the class of (vector-valued) functions on the input domain $\mathcal{X} \subset \mathbb{R}^{d_x}$ of the form 
\begin{equation}
x \to \sigma  (W_{d}  \, \sigma ( \cdots \sigma (W_1 x)))
\end{equation}
\noindent where  each of the weight matrices $W_i \in \mathbb{R}^{m_i \times n_i}$ has a bounded Frobenius norm, that is, $\lVert W_i\rVert_F \leq R_i$ for some $R_i \geq 0$, $m_i$ and $n_i$ are even for all $i = 1,...,d$ except $n_1$, and $\sigma$ is the GroupSort function with group size 2. If $d=0$, we define $\tilde{H}_0$ as a hypothesis class that contains only the identity map on $\mathcal{X} $.
\end{definition}
Note that Definition~\ref{def:vectorvalued} is very similar to  Definition \ref{GNN}, but the last layer has been removed, so that the resultant hypothesis class is vector-valued. Subsequently, the following lemma provides a way  to ``peel off'' layers using the GroupSort function, which is the main tool used in the derivation of ERC for GroupSort NNs.

\begin{lemma}\label{peelingoff}
Let $\tilde{H}_d$ be a vector-valued hypothesis class defined in Definition \ref{def:vectorvalued}, with $d \geq 1$, and suppose that $\lVert W_d \rVert_F \leq R_d$. For any dataset with $m$ data points, we have the following inequality:
\begin{equation}\label{eq:peelingoff_result}
\underset{\epsilon}{\mathbb{E}}  \sup_{h \in \tilde{\mathcal{H}}_{d}  }   \norm{ \frac{1}{m} \sum_{i = 1}^m  \epsilon_i h(x_i) } \leq 2 R_d \, \underset{\epsilon}{\mathbb{E}} \sup_{ h \in \tilde{\mathcal{H}}_{d-1}   } \frac{1}{m}   \norm{\sum_{i=1}^m \epsilon_i  h(x_i) } 
\end{equation}
\end{lemma} 
\begin{proof}
Letting $w_1^T,w_2^T,\cdots, w_k^T$ represent the rows of $W_{d}$, we have
\begin{subequations} \label{eq:lemma:expand}
\begin{align}
\underset{\epsilon}{\mathbb{E}}  \sup_{h \in \tilde{\mathcal{H}}_{d}  }   \norm{ \frac{1}{m} \sum_{i = 1}^m  \epsilon_i h(x_i) } 
&= \underset{\epsilon}{\mathbb{E}}  \sup_{\lVert W_{d}\rVert_F \leq R_{d} , h \in \tilde{\mathcal{H}}_{d-1}  }   \norm{ \frac{1}{m} \sum_{i = 1}^m  \epsilon_i \sigma(W_{d} h (x_i) ) }  \\
&= \underset{\epsilon}{\mathbb{E}}  \sup_{\lVert W_{d}\rVert_F \leq R_{d} , h \in \tilde{\mathcal{H}}_{d-1} }   \norm{ \frac{1}{m} \sum_{i = 1}^m  \epsilon_i \sigma 
\begin{bmatrix}
w_1^T h(x_i ) \\
w_2^T h(x_i) \\
\vdots \\
w_k^T h(x_i) 
\end{bmatrix}  }
\end{align}
\end{subequations}
By expanding the components of $\sigma$ using the identities for the maximum and minimum in Eqs.~\ref{eq:max-min}, Eq.~\ref{eq:lemma:expand} can be written as

\begin{equation}\label{eq:lemma:expand_final}
\underset{\epsilon}{\mathbb{E}}  \sup_{\lVert W_{d}\rVert_F \leq R_{d} , h \in \tilde{\mathcal{H}}_{d-1}  }   \norm{ \frac{1}{2m} \sum_{i = 1}^m  \epsilon_i  
 \begin{bmatrix}
w_1^T h(x_i ) + w_2^T h(x_i) + |w_1^T h(x_i ) - w_2^T h(x_i)|  \\
w_1^T h(x_i ) + w_2^T h(x_i) - |w_1^T h(x_i ) - w_2^T h(x_i)|\\
\vdots \\
w_{k-1}^T h(x_i ) + w_{k}^T h(x_i) + |w_{k-1}^T h(x_i ) - w_{k}^T h(x_i)|  \\
w_{k-1}^T h(x_i ) + w_{k}^T h(x_i) - |w_{k-1}^T h(x_i ) - w_{k}^T h(x_i)|  
\end{bmatrix}  }
\end{equation}

We can bound Eq.~\ref{eq:lemma:expand_final} using the triangle inequality with $A_1$ and $A_2$ defined as follows.
\begin{subequations}
\begin{small}
\begin{align}
A_1 =& 
\underset{\epsilon}{\mathbb{E}}  \sup_{\lVert W_{d}\rVert_F \leq R_{d} , h \in \tilde{\mathcal{H}}_{d-1}  }   \norm{ \frac{1}{2m} \sum_{i = 1}^m  \epsilon_i  
 \begin{bmatrix}
w_1^T h(x_i )  \\
 w_2^T h(x_i)\\
\vdots \\
w_{k-1}^T h(x_i )  \\
w_{k}^T h(x_i )  
\end{bmatrix}  }+
\underset{\epsilon}{\mathbb{E}}  \sup_{\lVert W_{d}\rVert_F \leq R_{d} , h \in \tilde{\mathcal{H}}_{d-1} }   \norm{ \frac{1}{2m} \sum_{i = 1}^m  \epsilon_i  
 \begin{bmatrix}
w_2^T h(x_i )  \\
 w_1^T h(x_i)\\
\vdots \\
w_{k}^T h(x_i )  \\
w_{k-1}^T h(x_i )   
\end{bmatrix}  } \label{eq:A1initial}
\end{align}

\begin{equation}\label{eq:A2_def}
A_2 = \underset{\epsilon}{\mathbb{E}}  \sup_{\lVert W_{d}\rVert_F \leq R_{d} , h \in \tilde{\mathcal{H}}_{d-1}  }   \norm{ \frac{1}{2m} \sum_{i = 1}^m  \epsilon_i  
\begin{bmatrix}
|w_1^T h(x_i ) - w_2^T h(x_i)|  \\
 - |w_1^T h(x_i ) - w_2^T h(x_i)|\\
\vdots \\
 |w_{k-1}^T h(x_i ) - w_{k}^T h(x_i)|  \\
 - |w_{k-1}^T h(x_i ) - w_{k}^T h(x_i)|  
\end{bmatrix}   }
\end{equation}
\end{small}
\end{subequations}
We first bound $A_1$ appropriately. By noting that the first term in Eq.~\ref{eq:A1initial} is equal to the second term since we only permuted the components, the following equation is derived for $A_1$.
\begin{equation}\label{eq:A1}
A_1 = \underset{\epsilon}{\mathbb{E}}  \sup_{\lVert W_{d}\rVert_F \leq R_{d} , h \in \tilde{\mathcal{H}}_{d-1} }   \norm{ \frac{1}{m} \sum_{i = 1}^m  \epsilon_i  
 \begin{bmatrix}
w_1^T h(x_i )  \\
 w_2^T h(x_i)\\
\vdots \\
w_{k-1}^T h(x_i )  \\
w_{k-1}^T h(x_i )  
\end{bmatrix}  }  = \underset{\epsilon}{\mathbb{E}}  \sup_{\lVert W_{d}\rVert_F \leq R_{d} , h \in \tilde{\mathcal{H}}_{d-1}  }   \norm{ \frac{1}{m} \sum_{i = 1}^m  \epsilon_i W_{d} h (x_i) }
\end{equation}

The RHS of Eq.~\ref{eq:A1} can be bounded using similar strategies that can be found in Lemma 3.1 of \cite{golowich2018size}. Subsequently, we rewrite the expression inside the supremum and expectation of the first term as follows.

\begin{equation}\label{eq:inside_sup}
\sqrt{\sum_{j=1}^k \lVert w_j \rVert^2 \biggl( \sum_{i=1}^m \biggl(\epsilon_i  \frac{w_j^T}{\lVert w_j \rVert} h(x_i) \biggl) \biggl)^2 }
\end{equation}
Note that Eq.~\ref{eq:inside_sup} is maximized when one of the $\lVert w_j \rVert = R_d$ and the rest of $\lVert w_l \rVert = 0$ for $l \neq j$ (this is simply because we maximize a positive linear function over $\lVert w_1 \rVert^2  ,\lVert w_2 \rVert^2 , \cdots , \lVert w_h \rVert^2 $). Therefore, it follows that 
\begin{align} 
\underset{\epsilon}{\mathbb{E}}  \sup_{\lVert W_{d}\rVert_F \leq R_{d} , h \in \tilde{\mathcal{H}}_{d-1}   }   \norm{ \frac{1}{m} \sum_{i = 1}^m  \epsilon_i  
 \begin{bmatrix}
 \nonumber
w_1^T h(x_i )  \\ 
 w_2^T h(x_i)\\
\vdots \\
w_{k-1}^T h(x_i )  \\
w_{k-1}^T h(x_i )  
\end{bmatrix}  } 
&= \underset{\epsilon}{\mathbb{E}}  \sup_{\lVert w \rVert  \leq R_d , h \in \tilde{\mathcal{H}}_{d-1}  }    \frac{1}{m} \sum_{i = 1}^m  \epsilon_i  
 w^T h(x_i) \\ \nonumber
 &\leq \underset{\epsilon}{\mathbb{E}}  \sup_{\lVert w \rVert  \leq R_d , h \in \tilde{\mathcal{H}}_{d-1}  }    \frac{1}{m} \lVert w \rVert \bigg\| \sum_{i = 1}^m  \epsilon_i  
  h(x_i) \bigg\| \\ 
   &= R_d \, \underset{\epsilon}{\mathbb{E}}  \sup_{ h \in \tilde{\mathcal{H}}_{d-1}  }    \frac{1}{m} \bigg\| \sum_{i = 1}^m  \epsilon_i  
  h(x_i) \bigg\|
\end{align}
The second term $A_2$ can be rewritten in the following way. 

\begin{equation}\label{eq:A2}
A_2 = \underset{\epsilon}{\mathbb{E}}  \sup_{\lVert W_{d}\rVert_F \leq R_{d} , h \in \tilde{\mathcal{H}}_{d-1}   }   \norm{ \frac{1}{2m} \sum_{i = 1}^m  \epsilon_i  
\begin{bmatrix}
|w_1^T h(x_i ) + w_2^T h(x_i)|  \\
 |w_1^T h(x_i ) + w_2^T h(x_i)|\\
\vdots \\
 |w_{k-1}^T h(x_i ) + w_{k}^T h(x_i)|  \\
 |w_{k-1}^T h(x_i ) + w_{k}^T h(x_i)|  
\end{bmatrix}   }
\end{equation}
Note that the negative signs in even coordinates can be removed since we take the magnitude of the vector, and removal of negative signs in even rows $w_2, w_4, \cdots, w_k$ is possible by the symmetry of the set $\lVert W_{d}\rVert_F \leq R_{d}$. Then, we rewrite the inner expression of Eq.~\ref{eq:A2} as follows.
\begin{equation} \label{eq:inner}
\frac{1}{2m} \sqrt{\sum_{j=1}^{k/2} 2 \lVert w_{2j} + w_{2j-1}  \rVert^2 \biggl( \sum_{i=1}^m \epsilon_i  \bigg|  \frac{(w_{2j} + w_{2j-1})^T}{\lVert  w_{2j} + w_{2j-1} \rVert} h(x_i) \bigg|   \biggl)^2 } 
\end{equation}
Using the triangle inequality property of norms, Eq.~\ref{eq:inner} can be bounded by
\begin{equation}\label{eq:inner_bound}
\frac{1}{2m} \sqrt{\sum_{j=1}^{k/2} 4 \bigg( \lVert w_{2j} \rVert^2+ \lVert w_{2j-1}  \rVert^2 \biggl) \biggl( \sum_{i=1}^m \epsilon_i  \bigg|  \frac{(w_{2j} + w_{2j-1})^T}{\lVert  w_{2j} + w_{2j-1} \rVert} h(x_i) \bigg|   \biggl)^2 }
\end{equation}
It is readily shown that Eq.~\ref{eq:inner_bound} is maximized when $ \lVert w_{2j} \rVert^2+ \lVert w_{2j-1}  \rVert^2  = R_d^2$ for some $j$, by a similar argument to the one for $A_1$. Thus, Eq.~\ref{eq:inner_bound} can be further bounded by 
\begin{equation} \frac{1}{2m} \sqrt{ 4 R_d^2 \biggl( \sum_{i=1}^m \epsilon_i  \bigg|  \frac{(w_{2j} + w_{2j-1})^T}{\lVert  w_{2j} + w_{2j-1} \rVert} h(x_i) \bigg|   \biggl)^2 } = \frac{R_d}{m} \bigg| \sum_{i=1}^m \epsilon_i  \bigg|  \frac{(w_{1} + w_{2})^T}{\lVert  w_{1} + w_{2} \rVert} h(x_i) \bigg| \bigg|
\end{equation}
Therefore, we finally derive the bound for $A_2$ defined in Eq.~\ref{eq:A2_def} as follows.
\begin{small}
\begin{align}\nonumber
A_2 
&= \underset{\epsilon}{\mathbb{E}}  \sup_{\lVert W_{d}\rVert_F \leq R_{d} , h \in \tilde{\mathcal{H}}_{d-1}   }   \norm{ \frac{1}{2m} \sum_{i = 1}^m  \epsilon_i  
\begin{bmatrix}
|w_1^T h(x_i ) + w_2^T h(x_i)|  \\
 |w_1^T h(x_i ) + w_2^T h(x_i)|\\
\vdots \\
 |w_{k-1}^T h(x_i ) + w_{k}^T h(x_i)|  \\
 |w_{k-1}^T h(x_i ) + w_{k}^T h(x_i)|  
\end{bmatrix}   } \\ \nonumber
&\leq \, \underset{\epsilon}{\mathbb{E}}  \sup_{\lVert W_{d}\rVert_F \leq R_{d} , h \in \tilde{\mathcal{H}}_{d-1}   } \frac{ R_d}{m} \bigg| \sum_{i=1}^m \epsilon_i  \bigg|  \frac{(w_{1} + w_{2})^T}{\lVert  w_{1} + w_{2} \rVert} h(x_i) \bigg| \bigg| \\ \nonumber
&\leq R_d \, \underset{\epsilon}{\mathbb{E}}  \sup_{\lVert W_{d}\rVert_F \leq R_{d} , h \in \tilde{\mathcal{H}}_{d-1}   } \frac{ 1}{m} \bigg| \sum_{i=1}^m \epsilon_i   \frac{(w_{1} + w_{2})^T}{\lVert  w_{1} + w_{2} \rVert} h(x_i) \bigg|  \qquad \text{by  Talagrand's Contraction Lemma}\\ \nonumber
&\leq R_d \, \underset{\epsilon}{\mathbb{E}}  \sup_{\lVert W_{d}\rVert_F \leq R_{d} , h \in \tilde{\mathcal{H}}_{d-1}   } \frac{1}{m} \norm{\frac{(w_{1} + w_{2})^T}{\lVert  w_{1} + w_{2} \rVert} } \norm{ \sum_{i=1}^m \epsilon_i    h(x_i)} \\
&\leq R_d  \, \underset{\epsilon}{\mathbb{E}}  \sup_{ h \in \tilde{\mathcal{H}}_{d-1}   } \frac{1}{m} \norm{ \sum_{i=1}^m \epsilon_i    h(x_i)}
\end{align}
\end{small}
In the second inequality, a slightly different version of Talagrand's Contraction Lemma is used, which can be found in \cite{mohri2018foundations}. Combining $A_1$ and $A_2$, we derive Eq.~\ref{eq:peelingoff_result} as follows.
\begin{align}\nonumber
\underset{\epsilon}{\mathbb{E}}  \sup_{h \in \tilde{\mathcal{H}}_{d}  }   \norm{ \frac{1}{m} \sum_{i = 1}^m  \epsilon_i h(x_i) } 
&\leq A_1 + A_2 \\ \nonumber
&\leq  R_d  \, \underset{\epsilon}{\mathbb{E}}  \sup_{ h \in \tilde{\mathcal{H}}_{d-1}   } \frac{1}{m} \norm{ \sum_{i=1}^m \epsilon_i    h(x_i)} +  R_d  \, \underset{\epsilon}{\mathbb{E}}  \sup_{ h \in \tilde{\mathcal{H}}_{d-1}   } \frac{1}{m} \norm{ \sum_{i=1}^m \epsilon_i    h(x_i)}  \\
&= 2  R_d  \, \underset{\epsilon}{\mathbb{E}}  \sup_{ h \in \tilde{\mathcal{H}}_{d-1}   } \frac{1}{m} \norm{ \sum_{i=1}^m \epsilon_i    h(x_i)}
\end{align}
This completes the proof. 
\end{proof}

\begin{thm}\label{maintheorem}
Assume that $d  \geq 1$ and $\mathcal{H}_d$ is the hypothesis class of real-valued functions defined in Definition \ref{GNN}. Suppose that $\mathcal{X} \subset \mathbb{R}^{d_x}$ is a bounded subset such that for all $x \in \mathcal{X}, \lVert x \rVert \leq B$. For any set of $m$ training samples $S = \{x_1, x_2, \cdots , x_m \}$, we have
\begin{equation} \label{eq:maintheorem_result}
\mathcal{R}_S (\mathcal{H}_d)  \leq \frac{B}{\sqrt{m}} 2^{d-1} \Pi_{i=1}^d R_i  
\end{equation}
where $\lVert W_i \rVert_F \leq R_i$ for each weight matrix in $\mathcal{H}_d$. 
\end{thm}
\begin{proof}
Since the proof is very similar to that in \cite{golowich2018size}, we provide a proof sketch only for clarity. Note that 
\begin{equation}
    \mathcal{R}_S (\mathcal{H}_d) = \underset{\epsilon}{\mathbb{E}}  \sup_{h \in \mathcal{H}_d  } \frac{1}{m} \sum_{i = 1}^m  \epsilon_i h(x_i) = \underset{\epsilon}{\mathbb{E}}  \sup_{\lVert W_d\rVert_2 \leq  R_d , h \in \tilde{\mathcal{H}}_{d-1}  } \frac{1}{m} \sum_{i = 1}^m  \epsilon_i  W_d h(x_i)
\end{equation}
Since $W_d$ has only 1 row, it follows that $\lVert W_d \rVert_2  = \lVert W_d \rVert_F$ by Lemma \ref{fnorm}.
By the Cauchy-Schwarz inequality, the RHS of Eq.~\ref{eq:maintheorem_result} can be bounded by:

\begin{align} \nonumber
\underset{\epsilon}{\mathbb{E}}  \sup_{\lVert W_d\rVert_F \leq R_d , h \in \tilde{\mathcal{H}}_{d-1}}  \frac{1}{m} \sum_{i = 1}^m  \epsilon_i  W_d h(x_i) 
&\leq \underset{\epsilon}{\mathbb{E}}  \sup_{\lVert W_d\rVert_2  \leq R_d , h \in \tilde{\mathcal{H}}_{d-1}  } \lVert  W_d   \rVert_2  \norm{\frac{1}{m} \sum_{i = 1}^m  \epsilon_i h(x_i) } \\
&\leq R_d \, \underset{\epsilon}{\mathbb{E}}  \sup_{h \in \tilde{\mathcal{H}}_{d-1}  }   \norm{ \frac{1}{m} \sum_{i = 1}^m  \epsilon_i h(x_i) } 
\end{align}
By recursively applying Lemma \ref{peelingoff}, we derive the following inequality:
\begin{equation}
 R_d \, \underset{\epsilon}{\mathbb{E}}  \sup_{h \in \tilde{\mathcal{H}}_{d-1}  }   \norm{ \frac{1}{m} \sum_{i = 1}^m  \epsilon_i h(x_i) } \leq  (2^{d-1} \Pi_{i=1}^d R_i) \, \underset{\epsilon}{\mathbb{E}}     \norm{ \frac{1}{m} \sum_{i = 1}^m  \epsilon_i x_i }
 \end{equation}
The expected value can then be bounded using Jensen's inequality:

\begin{align} \nonumber
\underset{\epsilon}{\mathbb{E}}    \norm{ \frac{1}{m} \sum_{i = 1}^m  \epsilon_i x_i } 
& \leq  \frac{1}{m}\sqrt{\underset{\epsilon}{\mathbb{E}}   \norm{ \sum_{i = 1}^m  \epsilon_i x_i }^2 } \\ \nonumber
& = \frac{1}{m}\sqrt{\underset{\epsilon}{\mathbb{E}}   \sum_{i = 1}^m  \sum_{j = 1}^m \epsilon_i \epsilon_j x_j^T x_i } \\
&= \frac{1}{m}\sqrt{\underset{\epsilon}{\mathbb{E}} \,  m \lVert x \rVert^2 } 
=\frac{B}{\sqrt{m}}
\end{align}

Therefore,  the following inequality is derived, and this completes the proof. 
\begin{equation}
    \mathcal{R}_S (\mathcal{H}_d)  \leq \frac{B}{\sqrt{m}} 2^{d-1} \Pi_{i=1}^d R_i  
\end{equation}
\end{proof}
Theorem~\ref{maintheorem} develops the upper bound for the ERC of GroupSort NNs. Based on the results derived in Theorem~\ref{maintheorem}, we subsequently derive the bound for the ERC of LCNNs. We begin with a lemma that relates the Frobenius norm to the spectral norm.
\begin{lemma}\label{fnorm}
Given a weight matrix $W \in \mathbb{R}^{m \times n}$, the following inequality holds: 
\begin{equation} 
	\lVert W \rVert_2\leq \lVert W\rVert _F \leq \min(m,n)  \lVert W \rVert_2
\end{equation}
\end{lemma} 
The proof can be readily obtained based on the fact that $\lVert W\rVert_F$ is the norm of the vector consisting of all singular values of $W$, and the spectral norm is simply the largest singular value. Using this lemma,  we develop the following bound for $\mathcal{R}_S (\mathcal{H}^{SD}_d)$, where $\mathcal{H}^{SD}_d$ is the hypothesis class of LCNNs using SpectralDense layers of depth $d$.

\begin{corollary}\label{mostimportant}
Let $\mathcal{H}^{SD}_d$ be the real-valued function hypothesis class defined in Definition \ref{SNN} with  $d  \geq 1$. Suppose that $\mathcal{X} \subset \mathbb{R}^{d_x}$ is a bounded subset such that for all $x \in \mathcal{X}, \lVert x \rVert \leq B$. For any set of $m$ training samples $S = \{x_1, x_2, \cdots , x_m \}$, we have the following inequality:
\begin{equation} \label{eq:corollary_result}
	\mathcal{R}_S (\mathcal{H}^{SD}_d) \leq \frac{B}{\sqrt{m}} 2^{d-1} \Pi_{i=1}^d \min(m_i,n_i) = \frac{B}{\sqrt{m}} 2^{d-1} \Pi_{i=1}^{d-1} \min(m_i,n_i) \end{equation}
where $m_i$ and $n_i$ are the number of rows and columns of the $i^{th}$ weight matrix $W_i$.  
\end{corollary}
\begin{proof}
The proof can be readily obtained using the fact that $\lVert W_i \rVert_2  = 1$ for all functions in $\mathcal{H}^{SD}_d$ and Lemma \ref{fnorm} (i.e., $\lVert W_i \rVert_F \leq  \min(m_i,n_i) \lVert W_i \rVert_2  =  \min(m_i,n_i)$). Then, by substituting this inequality into Theorem \ref{maintheorem}, we derive the inequality in Eq.~\ref{eq:corollary_result}. The last equality immediately follows since $m_d = 1$, and this completes the proof.
\end{proof}

\begin{rem}
Note that the bound derived in Eq.~\ref{eq:corollary_result} is a completely size-dependent bound for the ERC of $\mathcal{H}^{SD}_d$. Therefore, once the architecture of the LCNNs using SpectralDense layers has been decided, the ERC of the set of LCNNs is bounded by a constant, which only depends on the neurons in each layer, and most importantly, not on the choice of weights in each weight matrix. 
\end{rem}

Subsequently, using the vector contraction inequality in Lemma \ref{lemma:VCI}, we derive the following corollary that generalizes the results to the multi-dimensional output case.
\begin{corollary}\label{cor:multio}
Suppose that $\mathcal{X} \subset \mathbb{R}^{d_x}$ is a bounded subset such that for all $x \in \mathcal{X}, \lVert x \rVert \leq B$. Let $\mathcal{G}$ be the hypothesis class of loss functions defined in Eq. \ref{eq:hypothesisclassoflossfunctions}, where $\mathcal{H}$ is the hypothesis class of functions from $\mathcal{X}$ to $\mathcal{Y} \subset \mathbb{R}^{d_y}$ with each component function $h_k \in \mathcal{H}^{SD}_d$ for $k =1,...,d_y$,  and $L : \mathcal{Y} \times \mathcal{Y} \to \mathbb{R}^{\geq 0} $ is a $L_r$-Lipschitz loss function that satisfies the properties in Section \ref{Assumptions}. Then we have 
\begin{equation}\label{eq:cor:multio_result}
	\mathcal{R}_S (\mathcal{G}) \leq \frac{\sqrt{2} d_y L_r B}{\sqrt{m}} 2^{d-1} \Pi_{i=1}^{d-1} \min(m_i,n_i) 
\end{equation}
where $m_i$ and $n_i$ are the number of rows and columns in $W_i$.  
\end{corollary}
\begin{proof}
The proof of Eq.~\ref{eq:cor:multio_result} can be readily obtained by applying Corollary \ref{mostimportant} to Eq.~\ref{VCI} in Lemma \ref{lemma:VCI} and using Eq.~\ref{VCI2} as follows.

\begin{align}
\nonumber
\mathcal{R}_S (\mathcal{G}) = \underset{\epsilon}{\mathbb{E}}  \, \sup_{h \in \mathcal{H} } \frac{1}{m} \sum_{i = 1}^m  \epsilon_i L( h(x_i) , y)  &\leq \sqrt{2} L_r  \underset{\epsilon}{\mathbb{E}}  \, \sup_{h \in \mathcal{H} } \frac{1}{m} \sum_{i = 1}^m   \sum_{k=1}^{d_y} \epsilon_{ik} h_k(x_i) \\ \nonumber
&\leq  \sqrt{2} L_r  \sum_{k=1}^{d_y} \underset{\epsilon}{\mathbb{E}}  \sup_{h_k \in \mathcal{H}^{SD}_d } \frac{1}{m} \sum_{i = 1}^m    \epsilon_{ik} h_k(x_i) \\ \nonumber
&= \sqrt{2} L_r  d_y \, \underset{\epsilon}{\mathbb{E}}  \sup_{h_k \in \mathcal{H}^{SD}_d } \frac{1}{m} \sum_{i = 1}^m    \epsilon_{ik} h_k(x_i) \\ 
&= \frac{\sqrt{2} d_y L_r B}{\sqrt{m}} 2^{d-1} \Pi_{i=1}^{d-1} \min(m_i,n_i)
\end{align}

\end{proof}

Finally,  we develop the following theorem for the generalization error bound of LCNNs using SpectralDense layers as an immediate result  of Corollary~\ref{cor:multio}.

\begin{thm}\label{thm:ineq}
Suppose that $\mathcal{X} \subset \mathbb{R}^{d_x}$ is a bounded subset such that for all $x \in \mathcal{X}, \lVert x \rVert \leq B$. Let $\mathcal{G}$ be the hypothesis class of loss functions defined in Eq. \ref{eq:hypothesisclassoflossfunctions}, where $\mathcal{H}$ is the hypothesis class of functions from $\mathcal{X}$ to $\mathcal{Y} \subset \mathbb{R}^{d_y}$ with each component function $h_k \in \mathcal{H}^{SD}_d$ for $k =1,...,d_y$,  and $L : \mathcal{Y} \times \mathcal{Y} \to \mathbb{R}^{\geq 0} $ is a $L_r$-Lipschitz loss function that satisfies the properties in Section \ref{Assumptions}.
For any set of $m$ i.i.d. training samples $S = \{s_1, s_2, \cdots , s_m \}$ with  $s_i = (x_i,y_i) $, $i \in [1,m]$ drawn from a probability distribution $\mathcal{D} \subset \mathcal{X} \times \mathcal{Y} $, then for any $\delta \in (0,1)$, the following upper bound holds with probability $1-\delta$ :
\begin{equation}\label{eq:final_thm_result}
\underset{(x,y) \sim \mathcal{D}}{E} \, [L(h(x),y)] 
\leq \frac{1}{m} \sum_{i=1}^m L(h(x_i),y_i) +  \frac{ \sqrt{2} d_y L_r B}{\sqrt{m}} 2^{d} \, \Pi_{i=1}^{d-1} \min(m_i,n_i) + 3M \sqrt{\frac{\log \frac{1}{\delta}}{2m}} 
\end{equation}
where $m_i$ and $n_i$ are the number of rows and columns in $W_i$.  
\end{thm}
\begin{proof}
    The proof of Eq.~\ref{eq:final_thm_result} follows immediately from substituting Eq.~\ref{eq:cor:multio_result} (i.e., the upper bound for $\mathcal{R}_S (\mathcal{G})$) from Corollary \ref{cor:multio} into Eq.~\ref{eq:thm3B} in Theorem \ref{mohri}.  
\end{proof}
\subsection{Comparison of Empirical Rademacher Complexity Bounds}
The bound in Corollary \ref{mostimportant} shows that the ERC of LCNNs using SpectralDense layers is simply bounded by a constant that depends on the size of the network (i.e., the number of neurons in each layer and the depth of the network). However, if we consider the ERC of the class of conventional FNNs with 1-Lipschitz activation functions of depth $d$, denoted by $\mathcal{R}_S (\mathcal{H}^{D}_d)$, we have the following bound by \cite{golowich2018size} :
\begin{equation}
    \mathcal{R}_S (\mathcal{H}^{D}_d) \leq  \frac{B (\sqrt{2 \log(2) d}) \, \Pi_{i=1}^d R_i}{\sqrt{m}}
\end{equation}
A slight advantage of this bound is that we have an $O(\sqrt{d})$ dependence on the depth, while the bound in Corollary \ref{mostimportant} has an exponential dependence on $d$. However, the main advantage of the bound in Corollary \ref{mostimportant} is that it is only dependent on the size of the network, i.e., the number of neurons in each layer. For conventional FNNs with no constraints on the norm of the weight matrices, even though the training error could be rendered sufficiently small, the term $\Pi_{i=1}^d R_i $ is not controlled, which could result in a sufficiently large ERC for the class $\mathcal{H}^{D}_d$, and subsequently a large generalization error. 
\par
Furthermore, given an FNN using conventional dense layers and an LCNN using SpectralDense layers of the same size parameters, we have demonstrated that the ERC of LCNNs using SpectralDense layers is bounded by a fixed constant, while there is no constraint on the ERC bound of the FNNs using dense layers. This is a tremendous advantage since all the parameters in Eq.~\ref{eq:final_thm_result} in Theorem \ref{thm:ineq} are fixed except the training sample size $m$ as long as our target function is Lipschitz continuous and the neural network architecture is fixed. This implies that there is a provably correct probabilistic guarantee that the generalization error is at most $O(\frac{1}{\sqrt{m}})$. Therefore, it is demonstrated that LCNNs using SpectralDense layers can not only approximate a wide variety of nonlinear functions (i.e., the universal approximation theorem in Section~\ref{sec:universal_app_thm}), but also
effectively mitigate the issue of over-fitting to data noise and exhibit better generalization properties, since LCNNs  improve the robustness against data noise as compared to conventional FNNs and are developed with reduced model complexity.

\section{Application of LCNNs to Predictive Control of a Chemical Process}\label{sec:application}
In order to show the robustness of LCNNs, we develop LCNNs using SpectralDense layers and incorporate the LCNNs into a model predictive controller (MPC).  We will demonstrate that LCNNs using SpectralDense layers are able to accurately model the nonlinear dynamics of a chemical process, and furthermore, they are able to prevent over-fitting in the presence of data noise.

\subsection{Chemical Process Description}
The chemical process in consideration, which was first developed by \cite{wu_machine_2019}, is a non-isothermal, well-mixed CSTR that contains an exothermic second-order irreversible reaction that converts molecule $A$  to molecule $B$. There is a feed stream of $A$ into the CSTR, as well as a jacket that either supplies heat or cools the CSTR. We denote $C_A$ to be the concentration of $A$ and $T$ the temperature of the CSTR. The CSTR dynamics can be modeled with the governing equations as follows: 
\begin{equation}\label{eq:CSTR:ODEs}
	\begin{split}
\frac{dC_A}{dt} &= \frac{F}{V} (C_{A0} - C_A) - k_0 e^{-\frac{E}{RT}} C_A^2 \\
 \frac{dT}{dt}& = \frac{F}{V} (T_0 - T) - \frac{\Delta H}{\rho_L C_p} k_0 e^{-\frac{E}{RT}} C_A^2 + \frac{Q}{\rho_L C_p V }
\end{split}
\end{equation}
\noindent where $F$ is the feed flowrate, $\Delta H$ is the molar enthalpy of reaction, $k_0$ is the rate constant, $R$ is the ideal gas constant,   $\rho_L$ is the fluid density,  $E$ is the activation energy, and $C_p$ is the specific heat capacity.  $C_{A0}$  is the feed flow concentration of $A$, and $Q$  is the rate at which heat is transferred to the CSTR. The values of the process parameters used in this work are omitted, as they are exactly the same as those in \cite{wu_machine_learning_2}. 
The CSTR of Eq.~\ref{eq:CSTR:ODEs} has an unstable steady-state $[C_{As}, T_{s}, C_{A0_s}, Q_s]  = [1.95 ~\text{mol / dm}^3, 402~ \text{K} , 4 ~\text{mol / dm}^{3}, 0~ \text{kJ / h} ]$. We use $x$ to denote the system  states: $x^T = [C_A -C_{As}, T - T_{s}]$ and $u$ to denote the manipulated inputs: $u^T = [C_{A0} -C_{A0_s}, Q - Q_{s}]$, so that the equilibrium point $(x_s, u_s)$ is located at the origin. 
Following the formulation presented in Section \ref{classofsystems}, after extensive numerical simulations, a Lyapunov function $V(x) := x^T P x$ with $ P = \left[\begin{array}{cc} 1060  & 22 \\ 22 & 0.52 \end{array}\right]$ was constructed. The region $\Omega_\rho$ with $\rho = 372$ is found via open-loop simulations by sweeping through many feasible initial conditions within the domain space $D$, such that for any initial state $x_0 \in \Omega_\rho$, the controller $\Phi$ renders the origin of the state space exponentially stable.

\subsection{Data Generation}\label{Datageneration}
In order to train the LCNNs, we conducted open-loop simulations of the CSTR dynamics shown in (\ref{eq:CSTR:ODEs}) using many different possible control actions to generate the required dataset. Specifically, we performed a sweep of all possible initial states $x_0 \in \Omega_\rho$ and control actions $u \in U$. The forward Euler method with step size $h_c = 10^{-5}~ \text{hr}$ was used to obtain the value of $\tilde{F}(x_0,u)$, that is, to deduce the state after $\Delta = 10^{-3} ~ \text{hr}$ time has passed. The inputs in the dataset will be all such possible pairs $(x_0,u)$ and the outputs will be $\tilde{F}(x_0,u)$. 
We collected 20000 input-output pairs in the dataset, and   split it into training (52.5 \%), validation (17.5 \%), and testing (30 \%)  datasets. Before the training process, the dataset was pre-processed appropriately by standard scalers to ensure that each variable has a variance of the same order of magnitude.

\subsection{Model Training With Noise-Free Data}
For the prediction model $\tilde{F}_{nn}$ that approximates the function $\tilde{F}$, we first demonstrate that in the absence of training data noise, the LCNNs using SpectralDense layers can capture the dynamics of $\tilde{F}_{nn}$ well in the operating region $\Omega_{\rho}$. Specifically, we used an LCNN using two SpectralDense layers with 40 neurons each, followed by a dense layer with weights that have absolute values bounded by 1.0 and linear activation functions as the final layer. The weight bound was implemented using the \textit{max\_norm} constraint function provided by Tensorflow. The neural network was trained using the Adam Optimizer package provided by Tensorflow. The loss function that was utilized to train the neural networks was the mean squared error (MSE), and the testing error for the LCNN reached $2.83 \times 10^{-5}$, which was considered sufficiently small using normalized training data.

\subsection{Lyapunov-based MPC}
Subsequently, we incorporate the LCNN model into the design of Lyapunov-based MPC (LMPC). The control actions in LMPC are implemented using the sample-and-hold method, where $\Delta$ is the sampling period. Let $N$ be a positive integer that represents the prediction horizon of the MPC. The state at time $t=k\Delta$  is denoted by $x_k$, and the LMPC using LCNN model corresponds to the optimization problem described below, following 
the notation in \cite{wu_machine_2019}:
\begin{subequations}\label{eq:LMPC}
\begin{align}
\mathcal{J} &= \min_{\tilde{u}} \sum_{i=1}^N L(\tilde{x}_{k+i} , u_{k+i-1}) \label{eq:LMPC:cost}\\
\text{s.t.} \: \:  & \tilde{x}_{t+1} = \tilde{F}_{nn} (\tilde{x}_t, u_t),~~\forall t \in [k, k+N) \label{eq:LMPC:model}\\
& u_t \in U, ~~ \forall t \in [k, k+N) \label{eq:LMPC:input}\\
 &\tilde{x}_k= x_k \label{eq:LMPC:initial}\\
&V(\tilde{F}_{nn} (x_{k}, u_k) ) \leq V(\tilde{F}_{nn} (x_{k}, \Phi(x_k)) ),~~ \text{ if } \: x_k \in \Omega_{\rho} \setminus   \Omega_{\rho_{nn}} \label{eq:LMPC:constraint1}\\
&V(x_{t}) \leq \rho_{nn} \:\: \forall t \in [k, k+N],~~ \text{if }  \: x_k \in \Omega_{\rho_{nn}} \label{eq:LMPC:constraint2}
\end{align}
\end{subequations} 
\noindent where  $\tilde{u} = [ u_k,u_{k+1}, u_{k+2}, ...., u_{k+N-1} ]$, and $\Omega_{\rho_{nn}}$ is a  much smaller sublevel set than $\Omega_{\rho}$. The state predicted by the LCNN model $\tilde{F}_{nn} (\tilde{x}_t, u_t)$ is represented by $\tilde{x}$. Eq.~\ref{eq:LMPC:input} describes the input constraints imposed, since we assume that $U$ is the set of all possible input constraints. The initial condition of the prediction model is obtained from the feedback state measurement at $t=k\Delta$, as shown in Eq.~\ref{eq:LMPC:initial}. The constraints of Eqs.~\ref{eq:LMPC:constraint1}-\ref{eq:LMPC:constraint2} guarantee closed-loop stability, i.e., it ensures that the state which is originally in the set $\Omega_{\rho}$ will eventually converge to the much smaller sublevel set $\Omega_{\rho_{nn}}$, provided that a sufficiently small $\Delta > 0$ is used such that the  controller $\Phi$ when applied with the sample-and-hold method still guarantees convergence to the origin. Note that the key difference between the LMPC of Eq.~\ref{eq:LMPC} in this work and the one in our previous work \cite{wu_machine_2019} is that the LCNN model, which is a type of feedforward neural network, is used as the underlying model for prediction of future states in Eq.~\ref{eq:LMPC:model} to predict the state one sampling time forward, while in \cite{wu_machine_2019}, a recurrent neural network was developed to predict the trajectory of future states within one sampling period. Therefore, the formulated objective function shown in Eq.~\ref{eq:LMPC:cost} and the constraints of Eqs.~\ref{eq:LMPC:constraint1}-\ref{eq:LMPC:constraint2} only account for the predicted states in the sampling instance. The above optimization problem is solved using IPOPT, which is a package for solving large-scale nonlinear and non-convex optimization problems. The LMPC of Eq.~\ref{eq:LMPC} for the CSTR example is designed with the following parameters: $\rho = 372$, $\rho_{nn} = 2$, $N = 2$, and $L(x,u) = x^T Q_1 x  + u^T Q_2 u $, where
$ Q_1 = \left[\begin{array}{cc} 6.25 \times 10^{-4}  & 0 \\ 0 & 1 \end{array}\right]$ and $Q_2 = \left[\begin{array}{cc}  0.01  & 0 \\ 0 & 4.0 \times 10^{-12} \end{array}\right]$.

\subsection{Closed-Loop Simulation Results}\label{closedloop}
 The closed-loop simulation results under LMPC using LCNN models (termed LCNN-LMPC) are shown in Figure \ref{fig:1a} and Figure \ref{fig:1b}, where the initial condition is $x = (72~ \text{K} \: , \: -1.65~\text{ kmol / m}^{3})$. The closed-loop simulation under LMPC using the first-principles model of Eq.~\ref{eq:CSTR:ODEs} is also carried out as the reference for comparison purposes. 
  As shown in Figure \ref{fig:1a}, the state trajectory when then LCNN is used overlaps with the trajectory when the first-principles model is used in MPC, suggesting that the neural network modeling error is sufficiently small such that the state can be driven to the equilibrium point under MPC. The control actions taken by the two models are similar as shown in Fig.~\ref{fig:1b}, where slight deviations and oscillations occur under the MPC using LCNN model. The oscillations in the manipulated input profile could be due to the following factors: 1) slight model discrepancies between the model and the first-principles model, and 2) the IPOPT software being trapped within local minima of the objective function, or a combination of the two factors (in fact, the first factor might inadvertently cause the second because of irregularities in the LMPC loss function $L$ when LCNN is used). Nevertheless, the LCNN is able to drive the state to the origin effectively, which shows that the LCNN can effectively model the nonlinear dynamics of the CSTR.

\begin{rem}  
The weight bound of 1.0 in the last layer was chosen to ensure that the class of functions that can be represented using this neural network architecture is sufficiently large to contain the target function. If a smaller weight bound is chosen, the target function might not be approximated well, since the Lipschitz constant of the network does not meet the one for the target function.
\end{rem}

\subsection{Robustness against Data Noise}
As discussed in Section~\ref{sec:generalization}, one of the most crucial advantages of LCNNs using SpectralDense layers is their robustness against data noise and potential over-fitting during the training process. For example, when the number of neurons per hidden layer is exceptionally large,  the neural network tends to overestimate the complexity of the problem and ultimately learns the data noise (see \cite{ke2008empirical} and \cite{sheela2013review} for more details on this phenomenon). Therefore, in this subsection, we will demonstrate that when Gaussian data noise is introduced into the training datasets, the LCNNs outperform the FNNs using conventional dense layers (termed ``Dense FNNs''). 

We followed the data generation process in Section \ref{Datageneration}, and added Gaussian noise with a standard deviation of 0.1 or 0.2 to the training dataset. 
To show the difference between SpectralDense LCNNs and conventional Dense FNNs in robustness against over-fitting, we trained SpectralDense LCNNs and Dense FNNs with the same set of hidden layer architectures. Specifically, we used two hidden layers in each type of neural network with 640 or 1280 neurons each. The LCNNs were developed with SpectralDense hidden layers, while the conventional Dense FNNs were developed using the dense layers from Tensorflow with ReLU activation functions. Throughout the training process for all the networks, the same training hyperparameters were used, such as the number of epochs, the early stopping callback, and the batch size used in the Adam Optimizer.
\par
Table \ref{Testing Error Comparison} shows the testing errors for the various neural networks trained. As seen in Table~\ref{Testing Error Comparison}, the testing errors of the conventional Dense FNNs have an order of magnitude of $10^{-3}$ to $10^{-2}$, which are significantly larger than those of SpectralDense LCNNs (i.e., the testing errors of SpectralDense LCNNs have an order of magnitude of $10^{-5}$ to $10^{-4}$). The increase in testing error in the Dense FNNs is due to over-fitting of the noise since the testing error has the same order of magnitude as the variance of the Gaussian noise. 
\par
Additionally, we integrated the LCNN and Dense FNN with 640 neurons per layer and 0.1 standard deviation Gaussian Noise into MPC, similar to the process described in Section \ref{closedloop}. The results when both the LCNNs and Dense FNNs are integrated into MPC are shown in Figure \ref{fig:2a} and \ref{fig:2b}.  From the plot of the Lyapunov function value $V(x)$ in Figure \ref{fig:2a}, it is demonstrated that the Dense FNN is unable to effectively drive the state to the origin compared to the LCNN. This is readily observed because not only is the Lyapunov function value for the Dense FNN much higher, but there are also considerable oscillations in the function value, especially in the time frame between $0.15$ hr and $0.3$ hr. In addition, in Figure \ref{fig:2b}, it is observed that, while the predicted control actions of the LCNN are very similar to those of the first-principles model, the predicted control actions under the Dense FNN show large oscillations and differ largely from those of the first-principles model. This large disparity between the Dense FNN and the first-principles model shows that the Dense FNN has become incapable of accurately modeling the process dynamics when embedded into MPC.

\subsection{Comparison of Lipschitz Constants between LCNNs and Dense FNNs}\label{lipschitz constant comparison}
Additionally, we compare the Lipschitz constants of the LCNNs and the Dense FNNs developed for the CSTR of Eq.~\ref{eq:CSTR:ODEs} and demonstrate that the conventional Dense FNNs have a much larger Lipschitz constant than the SpectralDense LCNNs as a result of noise over-fitting and the lack of constraints on the weight matrices. For the Dense FNNs, we used the LipBaB algorithm to obtain the Lipschitz constant for the FNNs using dense layers, and for the SpectralDense LCNNs, we took the SVD of the last weight matrix and obtained the spectral norm of the last weight matrix as the upper bound of the Lipschitz constant. The results are shown in Table \ref{LipschitzConstantComparison}. The Dense FNNs have Lipschitz constants in the order of magnitude $10^2$ compared to SpectralDense LCNNs with Lipschitz constants in the order of magnitude $10^0$. The comparison of Lipschitz constants demonstrates that the LCNNs are also provably and certainly less sensitive to input perturbations as compared to the Dense FNNs since they have a Lipschitz constant several orders of magnitude lower. Furthermore, the calculation of Lipschitz constants demonstrates that LCNNs are able to prevent over-fitting data noise by maintaining a small Lipschitz constant, while conventional Dense FNNs with a large number of neurons could over-fit data noise.
\par

\section{Conclusions}
In this work, we developed LCNNs for the general class of nonlinear systems, and discussed how LCNNs using SpectralDense layers can mitigate sensitivity issues and prevent over-fitting to noisy data from the perspectives of Lipschitz constants and generalization error.  Specifically, we first proved the universal approximation theorem for LCNNs using SpectralDense layers to demonstrate that LCNNs are capable of retaining expressive power for Lipschitz target functions despite having a small hypothesis class. Then, we derived the generalization error bound for SpectralDense LCNNs using the Rademacher complexity method. The above results provided the theoretical foundations to demonstrate that LCNNs can improve input sensitivity due to their constrained Lipschitz constant and generalize better to prevent over-fitting. Finally, LCNNs using SpectralDense layers were integrated into MPC and applied to a chemical reactor example. The simulations show that the LCNNs   effectively captured the process dynamics and outperformed conventional Dense FNNs in terms of smaller testing errors, higher prediction accuracy in MPC, and smaller Lipschitz constants in the presence of noisy training data.

\section{Acknowledgments}
 Financial support from the NUS Start-up grant  R-279-000-656-731 is gratefully acknowledged. 

\newpage
\bibliography{reference}

\begin{figure}[h]
\includegraphics[width = 18cm , height = 11cm]{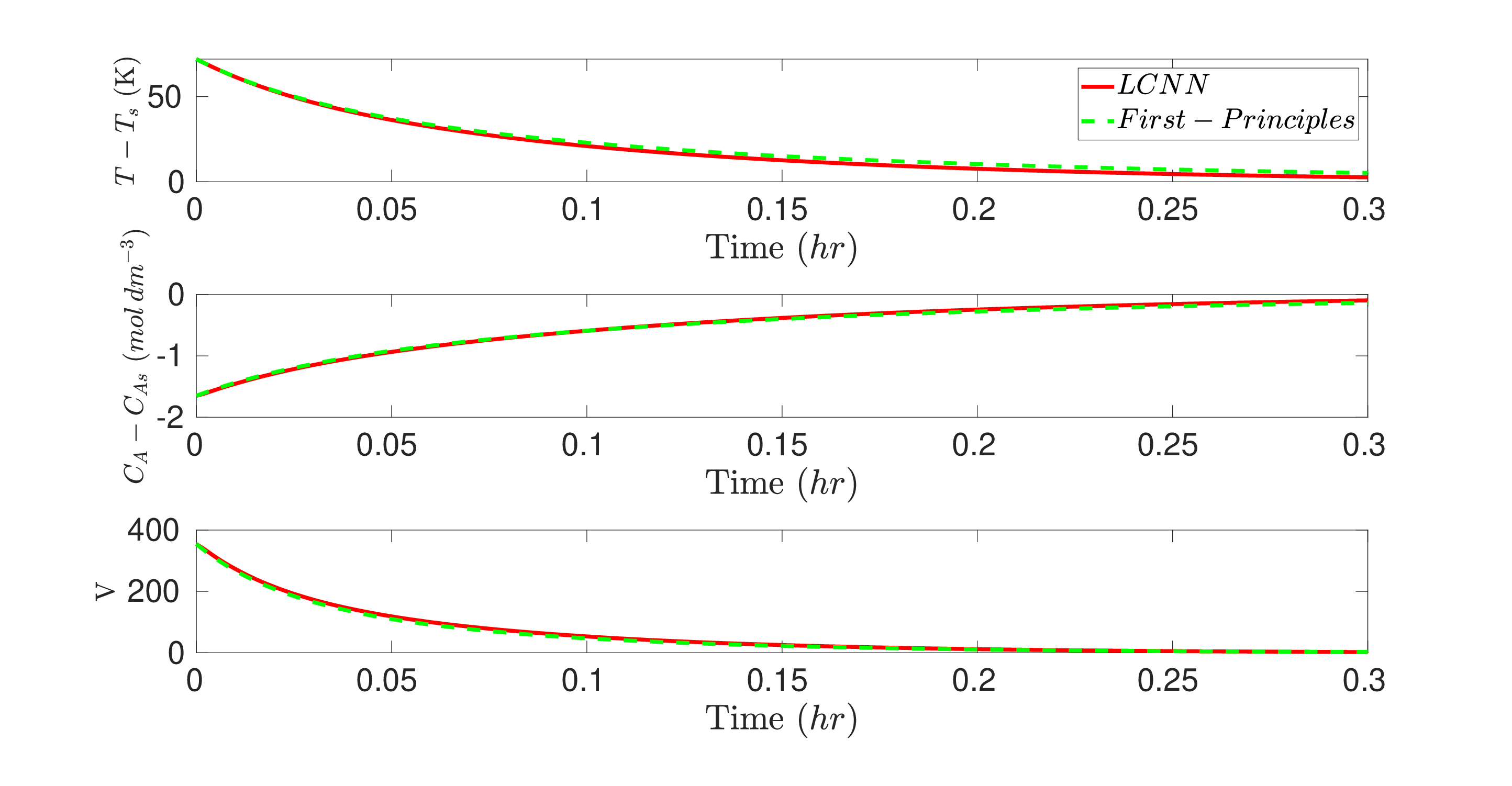}

\centering
\caption{Closed-loop state profiles of $T$ and $C_A$, and the evolution of the Lyapunov function $V(x)$ under the LMPCs using LCNN (red solid line) and first-principles model (green dashed line), respectively.}
\label{fig:1a}
\end{figure}
\begin{figure}[h]
\includegraphics[width = 18cm ]{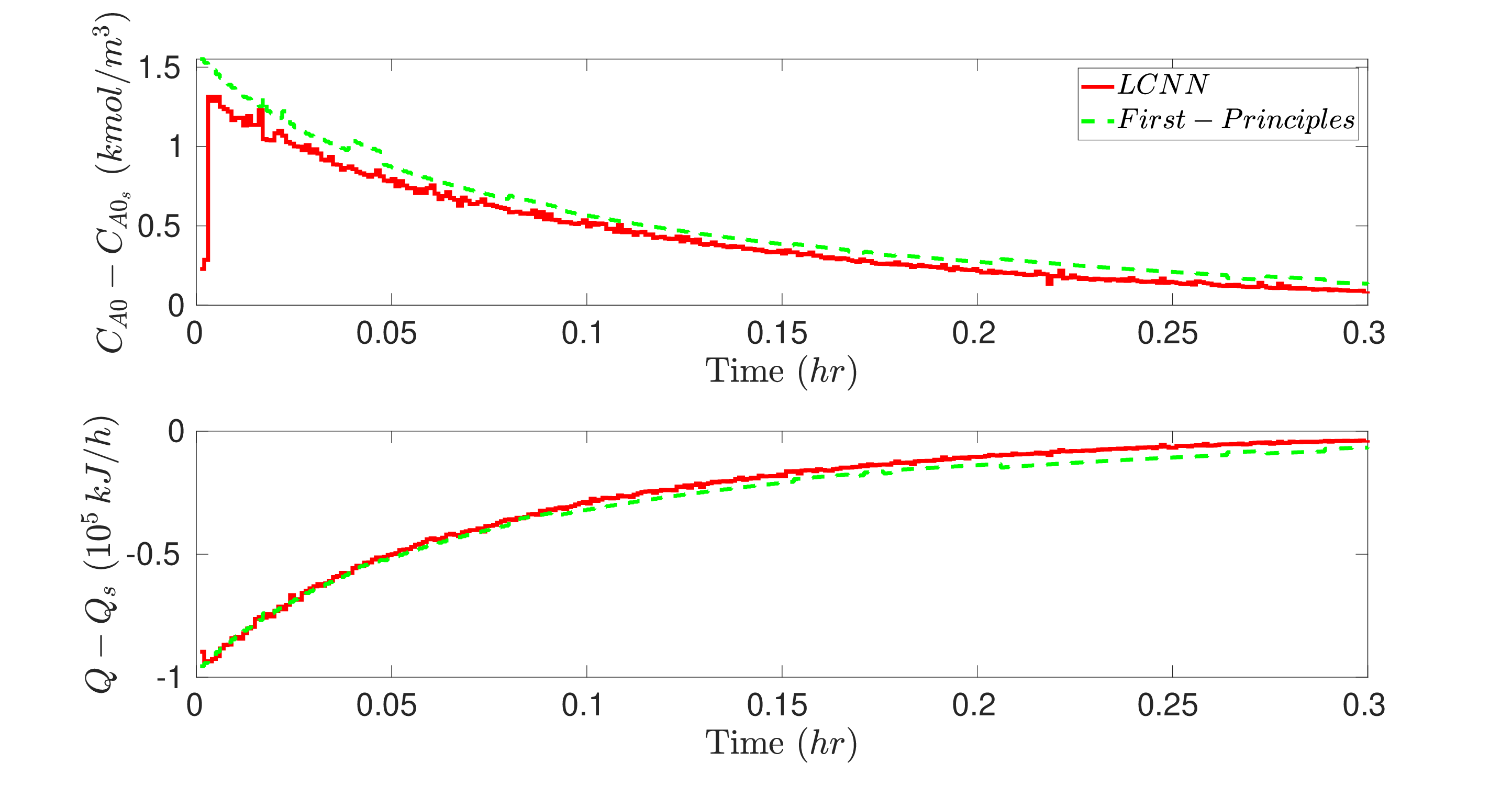}
\centering
\caption{Closed-loop manipulated input profiles of $C_{A0}$ and $Q$ under the LMPCs using LCNN (red solid line) and first-principles model (green dashed line), respectively.} 
\label{fig:1b}
\end{figure}

\begin{figure}[h]
\centering
\includegraphics[width = 18cm, height = 11cm ]{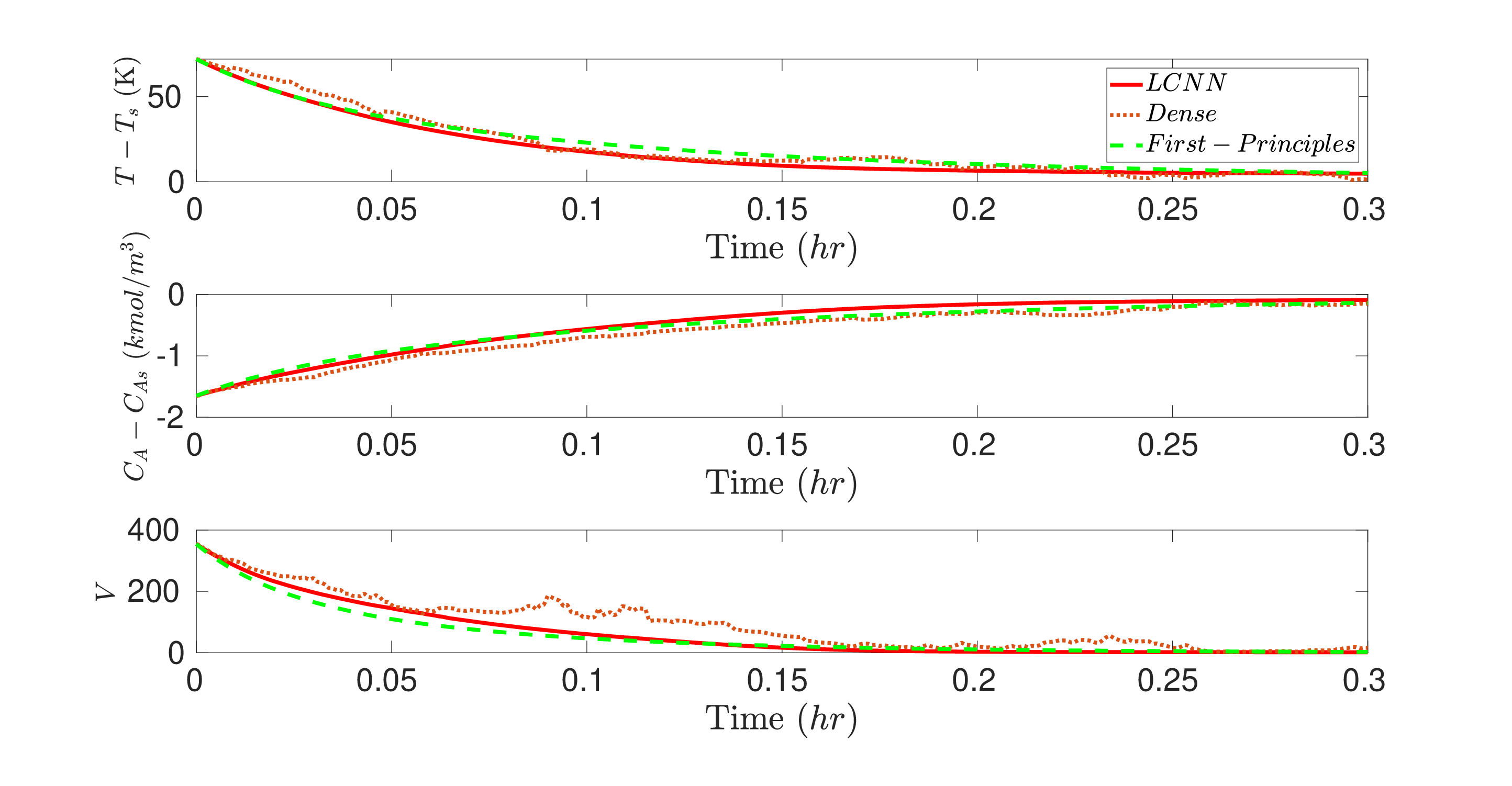}

\caption{Closed-loop state profiles of $T$ and $C_A$, and the evolution of the Lyapunov function $V(x)$ under the LMPCs using LCNN (red solid line), Dense FNN (orange dotted line), and first-principles model (green dashed line), respectively, where noisy training data is used to develop LCNN and Dense FNN.} 
\label{fig:2a}
\end{figure}
\begin{figure}[h]
\includegraphics[width = 18cm ]{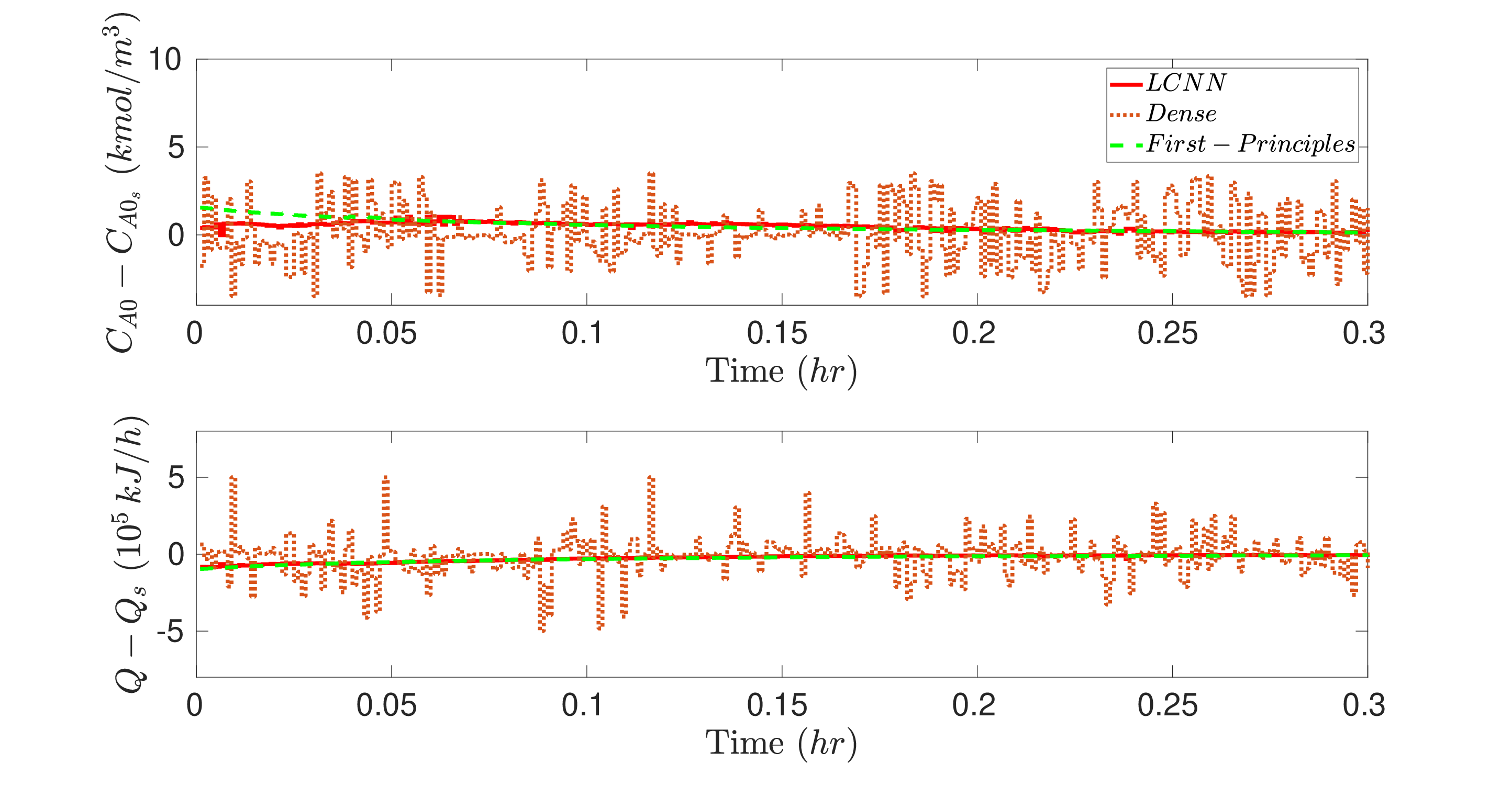}
\centering
\caption{Closed-loop manipulated input profiles of $C_{A0}$ and $Q$ under the LMPCs using LCNN (red solid line), Dense FNN (orange dotted line), and first-principles model (green dashed line), respectively, where noisy training data is used to develop LCNN and Dense FNN. }
\label{fig:2b}
\end{figure}

\begin{table}[h]
	\centering
	\begin{tabular}{llrrr}
		\toprule
		Hidden Layers & Noise SD &  LCNN Testing Error &  Dense Testing Error &  Improvement Factor \\
		\midrule
		(640,640) &   0.1 &                     $3.617 \times 10^{-5} $ &             $4.573 \times 10^{-3}$ &          $1.264 \times 10^{2}$ \\
		(1280,1280) &   0.1 &                     $3.850 \times 10^{-5}$ &             $6.928 \times 10^{-3}$ &          $1.799 \times 10^{2}$ \\
		(640,640) &   0.2 &                     $2.088\times 10^{-5}$ &             $1.692 \times 10^{-2}$ &           $8.102 \times 10^{1}$ \\
		(1280,1280) &   0.2 &                     $1.393\times 10^{-5}$ &             $3.095 \times 10^{-2}$ &          $ 2.220 \times 10^{2}$\\
		\bottomrule
	\end{tabular}

	\caption{Comparison of the testing errors for various hidden layer architectures and standard deviation (SD) of noise introduced into the training dataset. The improvement factor is the ratio between the two testing errors. }
	\label{Testing Error Comparison}
\end{table}

\begin{table}[h]
	\centering
	\begin{tabular}{llrr}
		\toprule
		Hidden Layers & Noise SD & Dense Lipschitz constant &  LCNN Lipschitz constant \\
		\midrule
		(640,640) &    0.1 &              $2.447\times 10^{2}$ &                          1.119 \\
		(1280,1280) &    0.1 &           $2.682\times 10^{2}$  &                          1.116\\
		(640,640) &    0.2 &        $2.511\times 10^{2}$  &                          1.114\\
		(1280,1280) &    0.2 &             $9.285 \times 10^{2}$  &                          1.107\\
		\bottomrule
	\end{tabular}
	\caption{Lipschitz constants for the various hidden layer architectures and standard deviation (SD) of noise introduced into the training dataset. }
	\label{LipschitzConstantComparison}
\end{table}

\end{document}